\documentclass[a4paper,numberwithinsect,cleveref,autoref,thm-restate]{lipics-v2021}
\usepackage{mathtools}
\usepackage{physics}

\DeclareMathOperator{\up}{\uparrow}
\DeclareMathOperator{\down}{\downarrow}
\DeclareMathOperator{\CSP}{CSP}

\DeclareMathOperator{\inc}{inc}
\DeclareMathOperator{\facts}{facts}

\newcommand{\If}{\hspace{6pt}\text{if}\hspace{6pt}}

\newcommand{\A}{\mathcal{A}}
\newcommand{\C}{\mathcal{C}}
\newcommand{\D}{\mathcal{D}}
\newcommand{\F}{\mathcal{F}}
\newcommand{\K}{\mathcal{K}}
\newcommand{\T}{\mathcal{T}}

\newcommand{\B}{\mathbb{B}}
\newcommand{\N}{\mathbb{N}}

\title{Adaptive Query Algorithms for Relational Structures Based on Homomorphism Counts}

\titlerunning{Adaptive Query Algorithms Based on Homomorphism Counts} 

\author{Balder ten Cate}{Institute for Logic, Language and Computation, University of Amsterdam, The Netherlands}{b.d.tencate@uva.nl}{https://orcid.org/0000-0002-2538-5846}{}

\author{Phokion G.~{Kolaitis}}{University of California Santa Cruz \& IBM Research - Almaden, USA}{kolaitis@soe.ucsc.edu}{https://orcid.org/0000-0002-8407-8563}{}

\author{Arnar Á. Kristjánsson}{Master of Logic, Institute for Logic, Language and Computation, University of Amsterdam, The Netherlands}{arnar.kristjansson@student.uva.nl}{https://orcid.org/0009-0002-4483-4667}{}

\authorrunning{B. {ten Cate}, P.G. {Kolaitis}, and A.Á. Kristjánsson} 

\Copyright{Balder {ten Cate}, Phokion G.~{Kolaitis}, and Arnar Á. Kristjánsson}



\nolinenumbers

\EventEditors{Pawe\l{} Gawrychowski, Filip Mazowiecki, and Micha\l{} Skrzypczak}
\EventNoEds{3}
\EventLongTitle{50th International Symposium on Mathematical Foundations of Computer Science (MFCS 2025)}
\EventShortTitle{MFCS 2025}
\EventAcronym{MFCS}
\EventYear{2025}
\EventDate{August 25--29, 2025}
\EventLocation{Warsaw, Poland}
\EventLogo{}
\SeriesVolume{345}
\ArticleNo{34}

\begin{document}

\maketitle

\begin{abstract}
A query algorithm based on homomorphism counts is a procedure to decide membership for a class of finite relational structures using only homomorphism count queries. 
A left query algorithm can ask the number of homomorphisms from any structure \emph{to} the input structure and a right query algorithm can ask the number of homomorphisms \emph{from} the input structure to any other structure.
We systematically compare the expressive power of different types of left or right query algorithms, including non-adaptive query algorithms, adaptive query algorithms that can ask a bounded number of queries, 
and adaptive query algorithms that can ask an unbounded number of queries.
We also consider query algorithms where the homomorphism counting is done over the Boolean semiring $\B$, meaning that only the existence of a homomorphism is recorded, not the precise number of them.

\end{abstract}

\section{Introduction}
A classic result by Lovász \cite{lovasz1967operations} 
states that a finite relational structure $A$ is determined up to isomorphism by the numbers $\hom(B, A)$ of homomorphisms from $B$ to $A$ for every finite relational structure $B$ with the same signature. This list of numbers is called the \emph{left homomorphism profile} of $A$ and is written as $\hom(\C, A)$, where $\C$ is the class of all finite relational structures with the same signature as $A$.
Similarly, one can define the \emph{right homomorphism profile} as the list of numbers of homomorphism from $A$ to each structure. The right profile also determines the structure up to isomorphism, as shown by Chaudhuri and Vardi \cite{chaudhuri1993optimization}. 
\par 
Recently, there has been considerable interest in exploring which structures can be distinguished by the homomorphism profiles
when the class $\C$ is restricted. As an example, in 2010, Dvo\v{r}ák \cite{dvovrak2010recognizing} proved that if $\C$ is the class of graphs of treewidth at most $k$, the left profile can distinguish precisely the same graphs as the
$k$-dimensional Weisfeiler-Leman algorithm, a well-known polynomial time graph isomorphism test. Cai, Fürer, and Immerman \cite{cai1992optimal} had already proved that the $k+1$ variable fragment of first-order logic extended with counting quantifiers can distinguish precisely the same graphs as $k$-dimensional Weisfeiler-Leman algorithm.
Also, in 2020, Grohe \cite{grohe2020counting} proved that first-order logic with counting of quantifier rank at most $k$ can distinguish exactly the same graphs as the left homomorphism profile restricted to graphs
of treedepth $k$. 
Man\v{c}inska and Roberson \cite{Mancinska2020:quantum} showed that the graphs that are quantum isomorphic are exactly the graphs that have the same left homomorphism profile restricted to planar graphs.
In 2024, Seppelt \cite{Seppelt2024:logical} provides further insight into the logical equivalence relations for graphs that can be characterized by restricted left homomorphism profiles, using techniques from structural graph theory.
These results highlight the importance of homomorphism profiles in finite model theory by describing a close connection between homomorphism profiles and logical languages. \looseness=-1

Taking a slightly different approach, in 2022, Chen et
al.~\cite{chen_et_al:LIPIcs.MFCS.2022.32} ask what properties of graphs can be decided based on 
\emph{finitely many} homomorphism count queries. They did this by introducing the notion of a homomorphism-count based \emph{query algorithm}. 
A \emph{non-adaptive left $k$-query algorithm} (also referred to as a left $k$-query algorithm) consists of a finite tuple $\mathcal{F}= (F_1, \ldots, F_k)$ of structures and a subset $X \subseteq \N^{k}$. The algorithm is said to decide the class of structures 
\begin{align*}
    \mathcal{A} \coloneqq \{ A : (\hom(F_i, A))_{i=1}^k \in X\}.
\end{align*}
Non-adaptive \emph{right} $k$-query algorithms
are defined similarly, but using
$\hom(A,F_i)$ instead of $\hom(F_i,A)$.
These notions allow for the study of what classes of structures can be defined using a finite restriction of the left (respectively, right) homomorphism profile. 
Chen et al.\ also define an \emph{adaptive} version of left/right $k$-query algorithms. These are, roughly speaking,  algorithms that can ask $k$ homomorphism-count queries where the choice of the $(i+1)$-th query may depend on the answers to the first $i$ queries.
Chen et al.~\cite{chen_et_al:LIPIcs.MFCS.2022.32}
studied the expressive power of non-adaptive and adaptive query algorithms for the specific case of \emph{simple graphs} (i.e., undirected graphs without self-loops).
Among other things, they showed that every first-order sentence $\phi$ that is a Boolean combination of universal first-order sentences can be decided by a non-adaptive $k$-left query algorithm for some $k$. On the
other hand, they showed that there are first-order sentences $\phi$ that cannot be decided by a non-adaptive left $k$-query algorithm for any $k$.
When it comes to adaptive query algorithms, they showed that every class of simple graphs can be decided by an adaptive left 3-query algorithm, 
but that there exists a class of simple graphs that is not decided by an adaptive right $k$-query algorithm for any $k$, when using only simple graphs for the queries. 

In a 2024 paper, ten Cate et al.~\cite{ten2024homomorphism} extend the above analysis by exploring query algorithms for arbitrary relational structures (not only simple graphs) and by considering also query algorithms over the Boolean semiring $\B$, which means that instead of being able to query for the number of homomorphisms, the algorithm can only query for the existence of a homomorphism. The query algorithms as defined by Chen et al.~\cite{chen_et_al:LIPIcs.MFCS.2022.32} can 
be viewed as query algorithms over the 
semiring $\N$ of natural numbers. It is 
shown in~\cite{ten2024homomorphism} that, 
for classes of structures that are closed under
homomorphic equivalence, non-adaptive left $k$-query algorithms (over $\N$) are no more
powerful than non-adaptive left $k$-query algorithms over $\B$. In other words, ``counting does not help'' for such classes. Moreover, for such classes, non-adaptive left query algorithms are equivalent in expressive power to first-order logic, meaning that such classes are first-order definable if and only if they are decided by a non-adaptive left query algorithm over $\N$ (or $\B$).

\begin{figure}
    \[\begin{array}{cccccccccccc}
    \mathfrak{A}_1 & \subsetneq & \mathfrak{A}_2 & \subsetneq & \cdots & \mathfrak{A}_k & \cdots &\subsetneq & \bigcup_k\mathfrak{A}_k & \subsetneq & \mathfrak{A}_{unb} = \text{All classes of structures}\\[2mm]
    \rotatebox{90}{$=$}  && \rotatebox{90}{$\subsetneq$} &&& \rotatebox{90}{$\subsetneq$} &&& 
    \rotatebox{90}{$\subsetneq$}\\[2mm]
    \mathfrak{N}_1 & \subsetneq & \mathfrak{N}_2 & \subsetneq & \cdots & \mathfrak{N}_k & \cdots &\subsetneq&
    \bigcup_k\mathfrak{N}_k 
    \end{array}\]
    \caption{Summary of our results regarding the relative expressive power of different types of left query algorithms over $\mathbb{N}$.
    Here, $\mathfrak{N}_k$, $\mathfrak{A}k$, and $\mathfrak{A}{unb}$ denote the collections of classes of structures that admit a non-adaptive left $k$-query algorithm, an adaptive left $k$-query algorithm, and an adaptive left unbounded query algorithm, respectively.
    No inclusions hold other than those depicted. In particular, $\mathfrak{A}_2\not\subseteq \bigcup_k\mathfrak{N}_k$ and $\mathfrak{N}_{k+1} \nsubseteq \mathfrak A_k$.}
    \label{fig:summary-left-N} 
\end{figure}

\subparagraph*{Our contributions.}
We study the expressive power of different versions of these query algorithms for classes of relational structures. In doing so, we characterize their expressive power, relate it to that of logics and database query languages, and compare it across different types of algorithms.
We focus mainly on adaptive query algorithms, including the previously unexplored notion of adaptive \emph{unbounded} query algorithms---i.e., adaptive query algorithms for which there is no uniform bound on the number of queries they can ask on a given input, but which nevertheless terminate on every input. 

\par 
Our main contributions are along
two lines: 

\begin{enumerate}
    \item 
We exploit the symmetry of directed cycles and cycle-like structures to give a simple formula (Lemma \ref{hom-lemma}) for the number of homomorphisms from a structure to a disjoint union of cycle-like structures of the same size. 
This formula serves as a powerful tool for deriving homomorphism indistinguishability results for finite profiles, and may be useful more broadly in related contexts.
We use it to systematically investigate and compare the expressive power of adaptive and non-adaptive query algorithms over $\N$ for relational structures.
For instance, we show, in Theorem \ref{better-corollary}, that for every signature containing a non-unary relation, there exists a class of structures of that signature that is not decided by an adaptive left $k$-query algorithm over $\N$ for any $k$. 
A concrete example of such a class is the class of directed graphs whose number of connected components is an even power of 2. This result is in sharp contrast to the result of Chen et al., which states that every class of simple graphs can be decided by an adaptive left 3-query algorithm over $\N$.
We also show in Theorem \ref{adaptive-not-better} that, for each $k>1$, there exists a class that is decided by a non-adaptive left $k$-query algorithm over $\N$ but not by any adaptive left $(k-1)$-query algorithm over $\N$.
However, in Theorem \ref{A2 not non-adaptive} we also note that there exists a class that is decided by an adaptive left 2-query algorithm but not by any non-adaptive left query algorithm. We thus establish all inclusion and non-inclusion relations between the expressive powers of adaptive and non-adaptive left $k$-query algorithms over $\N$ (note: it follows from the proof of Lov\'{a}sz's theorem that every class is decided by an adaptive left unbounded query algorithm). The results are summarized in Figure~\ref{fig:summary-left-N}.

\smallskip
\par 
Interestingly, the situation is quite different
for \emph{right} query algorithms over $\N$.
Indeed, it follows from results in \cite{wu2023study} that
every class of structures can be decided with only two adaptive right queries.
This is again in contrast to a prior result of Chen et al.~\cite{chen_et_al:LIPIcs.MFCS.2022.32}, namely that there is no $k$ such that every class of simple graphs can be decided by an adaptive right $k$-query algorithm (when the queries themselves are also required to be simple graphs).

\medskip 
\par
\item
We investigate the expressive power of adaptive query algorithms over $\B$.
We observe that adaptive \emph{bounded} query algorithms over $\B$ have the same expressive power as \emph{non-adaptive} query algorithms over $\B$, whose expressive power was characterized by ten Cate et al.~\cite{ten2024homomorphism}.
However, we show that adaptive unbounded query algorithms over $\B$ are strictly more expressive. 
The classes decided by such algorithms exhibit an intricate pattern, which we show can be captured by a topological characterization.
In the case of homomorphism-closed classes, we further show that the characterization can be simplified significantly, yielding an elegant condition for establishing whether such a class is decided by an adaptive unbounded query algorithm over $\B$ or not.
\smallskip
\par 
We also apply the topological characterization to examine the expressive power in the special cases of classes that are homomorphism equivalence types, CSPs, or Datalog-definable.
Specifically, we show that unboundedness does not increase the expressive power of adaptive left query algorithms over $\B$ for deciding CSPs and homomorphic-equivalence types, thereby yielding, based on existing results, an effective criterion for determining whether such classes are decided by an adaptive left unbounded query algorithm over $\B$.
Additionally, we characterize the Datalog definable classes that are decided by such an algorithm in terms of their upper envelopes.
Finally, we show that analogous results hold for adaptive \emph{right} unbounded query algorithms over $\B$. 
\end{enumerate}

\subparagraph*{Outline.}
Section~\ref{sec:preliminaries} reviews preliminary definitions and background, including those of relational structures, digraphs, and walks in digraphs.
Section~\ref{query-algs} introduces the main concepts studied in this paper. We review the definition of non-adaptive query algorithms and introduce a new definition of adaptive query algorithms that applies in both the bounded and unbounded settings.
Section~\ref{left-algs-N} contains our main results regarding the expressive power of query algorithms over $\N$.
Section~\ref{sec:unb helps} begins with the statement of a key lemma and then builds towards the result that not all classes can be expressed by an adaptive left bounded query algorithm over $\N$.
In Section~\ref{sec:adaptive vs non-adaptive}, we establish the results comparing adaptive and non-adaptive algorithms (as depicted in Figure 1). 
Section~\ref{right-algs-over-N} concisely describes the simpler situation for adaptive right query algorithms over $\N$.
Section~\ref{algs-over-B} presents our results regarding the expressive power of query algorithms over $\B$. 
In Section~\ref{sec:unb helps B}, we show that unboundedness increases the expressive power of adaptive left query algorithms over $\B$. 
Section~\ref{Unbounded qals over B} provides the topological characterization of the classes expressed by adaptive left unbounded query algorithms over $\B$, and in Section~\ref{sec:case studies} we apply this characterization to special cases.
In Section~\ref{sec:rqa over B}, we show that similar results can be obtained for adaptive right query algorithms over $\B$.
We conclude in Section~\ref{sec:discussion} by listing directions for further research. 
In the interest of space, we have moved the proofs of some theorems to the appendices.

\section{Preliminaries}
\label{sec:preliminaries}
\subparagraph*{Relational structures and homomorphism counts.}
A \emph{signature} $\tau = \{R_1, \ldots, R_n\}$ is a set of relational symbols where each symbol $R_i$ has an associated arity $r_i$. A \emph{relational structure} $\mathbf A$ of signature $\tau$ consists of a non-empty set $A$ and a relation $R_i^\mathbf A \subseteq A^{r_i}$ for each symbol $R_i \in \tau$. We use $\mathsf{FIN}(\tau)$ to denote the class of all finite relational structures with signature $\tau$.
We will only study finite relational structures; we will thus write ``structure'' and mean ``finite relational structure'' throughout this text.
If $\mathbf A = (A, (R_i^\mathbf A)_{i\in \{1,\ldots,n\}})$ and $\mathbf B = (B, (R_i^\mathbf B)_{i\in \{1,\ldots,n\}})$ are structures with signature $\tau$, we say that $\varphi: A \to B$ is a \emph{homomorphism} from $\mathbf A$ to $\mathbf B$ if it preserves all the relations, i.e. if \begin{align*}
   (a_1,\ldots, a_{r_i}) \in R_i^\mathbf A \Longrightarrow
   (\varphi(a_1),\ldots, \varphi(a_{r_i})) \in R_i^\mathbf B
\end{align*}
for each $R_i \in \tau$.
We then write $\varphi: \mathbf A \to \mathbf B$.
If there exists a homomorphism from $\mathbf A$ to $\mathbf B$ we write $\mathbf A \to \mathbf B$, otherwise, $\mathbf A \nrightarrow \mathbf B$. If $\mathbf A \to \mathbf B$ and $\mathbf B \to \mathbf A$, we then write $\mathbf A \leftrightarrow \mathbf B$ and say that $\mathbf A$ and $\mathbf B$ are homomorphically equivalent.
\par 

For structures $A$ and $B$, we let $\hom(A,B)$ denote the number of homomorphisms from $A$ to $B$. We also use $\hom_\N(A,B)$ to denote this same number. We then define \begin{align*}
    \hom_\B(A, B) \coloneqq \begin{cases}
        0 &\If \hom(A, B) = 0\\
        1 &\If \hom(A,B)>0
    \end{cases}
\end{align*}
where $\B$ denotes the Boolean semiring.
For a class $\C$ of structures, the \emph{left homomorphism profile of $A$ restricted to $\C$} is the tuple $(\hom(C, A))_{C \in \C}$. Similarly, we can define the right homomorphism profile and the homomorphism profile over $\B$.
When we speak of a class of structures, we will always assume it is closed under isomorphism.

\subparagraph*{Properties of structures.}
The incidence multigraph $\inc(\mathbf A)$ of a structure $\mathbf A = (A, (R^\mathbf A)_{R \in \tau})$ is the bipartite multigraph whose parts are the sets $A$ and $\facts(\mathbf A) = \{(R,\mathbf t) : R \in \tau \text{ and } \mathbf t \in R^{\mathbf A}\}$, and there is an edge between $a$ and $(R,\mathbf t)$ for each entry in $\mathbf t$ that is $a$.
We say that $\mathbf A$ is connected if $\inc (\mathbf A)$ is connected, and we say $\mathbf A$ is acyclic if $\inc \mathbf A$ is acyclic. In particular, if an element appears multiple times in a fact, then the structure is \emph{not} acyclic. This notion of acyclicity is also known as \emph{Berge-acyclicity}.
The number $c(\mathbf A)$ of components of $\mathbf A$ is the maximal $n$ such that $A$ can be written as the disjoint union of $n$ structures.
\par 

For given structures $A$ and $B$ we let $A \oplus B$ denote their disjoint union and for a natural number $m$ and a structure $H$ we write 
\begin{align*}
    m \cdot H \coloneqq \underbrace{H\oplus\ldots\oplus H}_{m-\text{times}}.
\end{align*}

\subparagraph*{Digraphs, walks, and homomorphisms to cycles.}
A \emph{digraph} (or a \emph{directed graph}) is a structure with exactly one binary relation. So it is a pair $(V, R)$ where $R \subseteq V \times V$.
A \emph{simple graph} is an undirected graph without loops.
We let $\mathsf{C}_n$ denote the directed cycle of length $n$:
\begin{align*}
    \mathsf{C}_n \coloneqq ( \{0,\ldots, n-1\}, \{(0,1), \ldots, (n-2,n-1), (n-1,0)\})
\end{align*}
and $\mathsf{P}_n$ denote the directed path of length $n$:\begin{align*}
    \mathsf{P}_n \coloneqq ( \{0,\ldots, n\}, \{(0,1), \ldots, (n-1,n)\})
\end{align*}
\par 
For our discussion of adaptive left query algorithms, walks in directed graphs play a large role. We will thus review the definition of an oriented walk and some related concepts.

\begin{definition} Let $A= (V, R)$ be a digraph.
\begin{enumerate}[(i)]
    \item An \emph{oriented walk} in $A$ is a sequence $(a_0, r_0, a_1, r_1, \ldots, a_n, r_n, a_{n+1})$ where $a_i \in V$ for each $i \in \{0,\ldots, n\}$, $r_i \in \{-,+\}$ and if $r_i = +$ then $(a_i, a_{i+1}) \in R$ while if $r_i = -$ then $(a_{i+1}, a_{i}) \in R$.
    \item The \emph{net length} of an oriented walk $W = (a_0, r_0, a_1, r_1, \ldots, a_n, r_{n}, a_{n+1})$ is defined as: \begin{align*}
        l(W) \coloneqq 
        \abs{ \{i : i \in \{0,\ldots, n\} \hspace{6pt}\text{and}\hspace{6pt} r_i = +\}} - \abs{\{i : i \in \{0,\ldots, n\} \hspace{6pt}\text{and}\hspace{6pt} r_i = -\}}.
    \end{align*}
    
    \item The oriented walk $(a_0, r_0, a_1, r_1, \ldots, a_n, r_{n}, a_{n+1})$ is said to be \emph{closed} if $a_0 = a_{n+1}$.

    \item A closed oriented walk $(a_0, r_0, a_1, r_1, \ldots, a_n, r_{n}, a_{0})$ is called an oriented \emph{cycle} if $a_i \neq a_j$ for all $i\neq j$ and $n\geq 0$.
    \item A \emph{directed walk} is an oriented walk $(a_0, r_0, \ldots, a_{n})$ such that $r_i = +$ for each $i$. \emph{Closed directed walks} and \emph{directed cycles} are defined analogously.
\end{enumerate}
\end{definition}

An important thing to note is that oriented walks and their net lengths are preserved under homomorphisms. So if $(a_0, r_0, a_1, r_1, \ldots, a_n, r_n, a_{n+1})$ is an oriented walk in $A$ and $\varphi: A \to B$ is a homomorphism then \[(\varphi(a_0), r_0, \varphi(a_1), r_1, \ldots, \varphi(a_n), r_n, \varphi(a_{n+1}))\] is an oriented walk in $B$ with the same net length.

For our proof of the existence of a class that is not decided by an adaptive left unbounded query algorithm, we take advantage of the fact that the existence of a homomorphism to a directed cycle can be described in an easy way.

\begin{lemma}[{Corollary 1.17 in \cite{hell2004graphs}}]
\label{girth-lemma}
    Let $A$ be a digraph and $n$ be a positive natural number. We have $A \to \mathsf{C}_n$ if and only if $n$ divides the net length of every closed oriented walk in $A$.
\end{lemma}

We use $\gcd$ to denote the greatest common divisor of a collection of natural numbers.
For every finite digraph, we can define a parameter very related to the condition in Lemma \ref{girth-lemma}.
\begin{definition}
    For a finite digraph $A$ we define \begin{align*}
        \gamma(A) \coloneqq \gcd(l_1, \ldots, l_k)
    \end{align*}
    where $l_1, \ldots, l_k$ is a listing of the net lengths of all cycles of positive net length in $A$. 
\end{definition}
Note, here we use the convention $\gcd(\varnothing) = 0$. 
For integers $n,m$ we write $n \mid m$ if $n$ divides $m$, and $n\nmid m$ otherwise.
The next proposition relates this parameter to the condition in Lemma \ref{girth-lemma}.
\begin{restatable}{proposition}{girthProp}
    \label{girth-prop}
    Let $A$ be a digraph and $n$ be a positive integer. We have $A \to \mathsf{C}_n$ if and only if $n \mid \gamma(A)$.
\end{restatable}

\subparagraph*{Conjunctive queries.}
A Boolean conjunctive query (CQ) is a first-order logic (FO)  sentence of the form $q=\exists \textbf{x}(\bigwedge_i \alpha_i)$ such that each $\alpha_i$ is of the form $R(y_1, \ldots, y_n)$ where each $y_i$ is among the variables in $\mathbf{x}$. It is well known that every such CQ has an associated ``canonical structure'' $A_q$ such that for every structure $B$ we have $B\vDash q$ if and only if $A_q \to B$.
Similarly,
vice versa, for every (finite) structure $A$, there is a 
Boolean CQ $q_A$ such that for every structure $B$ we have $A \to B$ if and only if $B\vDash q_A$.
A CQ is said to be Berge-acyclic if its canonical structure is acyclic.
A (Boolean) \emph{union of conjunctive queries (UCQ)} is a finite disjunction of (Boolean) CQs. Every Boolean UCQ is a positive existential FO sentence, and, conversely, every positive existential FO sentence is equivalent to a UCQ.

\section{Definitions of Query Algorithms}
\label{query-algs}
First, we review the definition of non-adaptive query algorithms.
\begin{definition}
\label{def:unbounded-adaptive-query-algs}
    Let $\C$ be a class of relational structures with the same signature, let $\K \in \{\N, \B\}$, and let $k$ be a positive integer.
    \begin{itemize}
        \item A \emph{non-adaptive left $k$-query algorithm} over $\K$ for $\C$ is a pair $(\F, X)$, where $ \F = (F_1, \ldots, F_k)$ is a tuple of relational structures with the same signature as the structures in $\C$, and $X$ is a set of $k$-tuples over $\K$, such that for all structures $D$, we have $D \in \C$ if and only if \[(\hom_\K(F_1, D), \hom_\K(F_2, D), \ldots, \hom_\K(F_k, D)) \in X\]
        \item A \emph{non-adaptive right $k$-query algorithm} over $\K$ for $\C$ is a pair $(\F, X)$, where $ \F = (F_1, \ldots, F_k)$ is a tuple of relational structures with the same signature as the structures in $\C$, and $X$ is a set of $k$-tuples over $K$, such that for all structures $D$, we have that $D \in \C$ if and only if \[(\hom_\K(D, F_1), \hom_\K(D, F_2), \ldots, \hom_\K(D, F_k)) \in X.\]
    \end{itemize}
\end{definition}
\smallskip 
\par 

\begin{example}
    Let $\F \coloneqq \{F\}$ where $F = (\{0\}, \varnothing)$ is the singleton digraph with no edges, and let $X = \{ 3 \}$. Then the non-adaptive left 1-query algorithm $(\F, X)$ over $\N$ decides the class $\C$ of digraphs that have exactly 3 vertices. It is not difficult to see that $\C$ does not admit a non-adaptive left query algorithm over $\B$. Indeed, this follows from the fact that the class is not closed under homomorphic equivalence. A query algorithm over $\B$ cannot distinguish between structures that are homomorphically equivalent, since if $A \leftrightarrow B$ then $D \to A$ if and only if $D \to B$ for any $D$.
\end{example}

The example above shows that query algorithms over $\N$ have more expressive power than query algorithms over $\B$ in general. However, for classes closed under homomorphic equivalence, the same does not hold, as shown by ten Cate et al. in the following theorem.
\begin{theorem}[Corollary 5.6 in \cite{ten2024homomorphism}]
    \label{Balder main theorem}
    Let $\C$ be a class of structures that is closed under homomorphic equivalence. Then the following are equivalent:
    \begin{itemize}
        \item $\C$ admits a non-adaptive left $k$-query algorithm over $\N$ for some $k$.
        \item $\C$ admits a non-adaptive left $k$-query algorithm over $\B$ for some $k$.
        \item $\C$ is FO-definable.
    \end{itemize}
\end{theorem}
\par 
We now formulate the definition of an \emph{adaptive unbounded query algorithm}.
Adaptive query algorithms were defined in Chen et al.~\cite{chen_et_al:LIPIcs.MFCS.2022.32}. The definition there is limited in the sense that it only considers query algorithms whose number of queries is bounded by some constant.

\noindent For a set $\Sigma$, we let $\Sigma^{< \omega}$ denote the set of finite strings with alphabet $\Sigma$.
\begin{itemize}
\item For $\sigma, \sigma' \in \Sigma^{<\omega}$, we write $\sigma' \sqsubseteq \sigma$ if $\sigma'$ is an initial segment of $\sigma$.
\item We let $\sigma \bullet \sigma'$ denote the concatenation of the strings $\sigma, \sigma'$ (we also use this notation if $\sigma'$ is an element of $\Sigma$).
\item If $\sigma_0 \sqsubseteq \sigma_1 \sqsubseteq \sigma_2  \sqsubseteq \ldots$ is a sequence, then we let $\bigsqcup_{n\geq0} \sigma_n$ denote the (possibly infinite) string $\sigma$ with length $m \coloneqq \sup_{n\geq 0} \abs{\sigma_n}$ such that for $i\leq m$, the $i$-th letter of $\sigma$ is the $i$-th letter of $\sigma_n$ for any $n$ such that $\abs{\sigma_n}\geq i$.
\end{itemize}
A set $\mathcal T \subseteq \Sigma^{<\omega}$ that is closed under initial segments (so if $\sigma\in \T$ and $\sigma'\sqsubseteq \sigma$ then $\sigma' \in \T$) is called a \emph{subtree of $\Sigma^{<\omega}$}. An element $\sigma$ of such a subtree $\mathcal T$ is called a \emph{leaf} if for every $\sigma' \in \T$ such that $\sigma \sqsubseteq \sigma'$ we have $\sigma' = \sigma$.
\par 
We are now ready to define adaptive left unbounded query algorithms.

\begin{definition} Let $\K \in \{\N, \B\}$.
\begin{enumerate}[(i)]
    \item An \emph{adaptive left unbounded query algorithm over $\K$} for structures with signature $\tau$ is a function 
    \[G: \mathcal{T} \to \mathsf{FIN}(\tau) \cup \{ \textsf{YES}, \textsf{NO}\}\] 
    where $\T$ is a subtree of $\K^{<\omega}$ and $G(\sigma) \in \{ \textsf{YES}, \textsf{NO}\}$ if and only if $\sigma$ is a leaf in $\T$.
    
    \item For $A \in \mathsf{FIN}(\tau)$, the \emph{computation path} of an adaptive left unbounded query algorithm $G$ on $A$ is the string defined by $\sigma(A, G) \coloneqq \bigsqcup_{n\geq 0} \sigma_n$ where \begin{align*}
        \sigma_0 &= \varepsilon \\
        \sigma_{n+1} &= \begin{cases}
            \sigma_n &\If G(\sigma_n) \in \{ \textsf{YES}, \textsf{NO}\} \\
            \sigma_n \bullet \hom_{\mathcal{K}}(G(\sigma_n), A) &\If G(\sigma_n) \in \mathsf{FIN}(\tau)
        \end{cases}.
    \end{align*}
    
    The algorithm $G$ \emph{halts} on input $A$ if $\sigma(A, G)$ is finite.
    
    \item If $G$ halts on all $A \in \mathsf{FIN}(\tau)$, then $G$ is said to be \emph{total}. 

    \item If $G$ is total, we say that it decides the class \begin{align*}
        \C \coloneqq \{ A \in \mathsf{FIN}(\tau) : G(\sigma(A,G)) = \mathsf{YES} \}.
    \end{align*}

    \item We say that $G$ is an adaptive left $k$-query algorithm if for every $A$, $\sigma(A, G)$ has length at most $k$. 
    We say $G$ is \emph{bounded} if it is an adaptive left $k$-query algorithm for some $k$.
    
\end{enumerate}

\end{definition}
The notion of an adaptive right unbounded/bounded query algorithm $G$ is defined analogously by appending $\hom_\mathcal{K}(A, G(\sigma_n))$ instead of $\hom_{\mathcal{K}}(G(\sigma_n), A)$ in the computation path.

\smallskip \par 
We will only consider total adaptive query algorithms, so when we speak of adaptive query algorithms we always assume it is total. 
\par 

\begin{example}
\label{example:2-adaptive-to-decide-directed-cycle}
Let $G$ be an adaptive left query algorithm over $\N$ for digraphs defined with the following:
\begin{align*}
    G(\sigma) \coloneqq \begin{cases}
        (\{0\}, \varnothing) &\If \sigma = \varepsilon\\
        \mathsf{P}_n &\If \sigma = `n\textrm'\\
        \mathsf{YES} &\If \sigma = `a_0a_1\textrm' \hspace{6pt}\text{and}\hspace{6pt}a_1 \neq 0\\
        \mathsf{NO} &\If \sigma = `a_0a_1\textrm' \hspace{6pt}\text{and}\hspace{6pt}a_1 = 0
    \end{cases}.
\end{align*}
The corresponding tree $\T\subseteq \N^{<\omega}$ would here be the set of all strings of length at most 2.\par 
We can describe a run of the algorithm on a target structure $A$ as follows:
First, the algorithm queries for the number of homomorphisms from the singleton digraph with no edges to $A$. Let $n$ denote the output of that query. 
It then queries for $\hom_\N(\mathsf{P}_n, A)$. If the result is positive, it accepts; otherwise, it rejects.
\par 
To see what class this algorithm decides, we first note that the number of homomorphisms from the singleton digraph with no edges to $A$ is precisely the number of vertices of $A$. We also have that there exists a homomorphism from $\mathsf{P}_n$ to $A$ if and only if $A$ has a directed walk of length $n$. Since $n$ is the size of $A$, we see that such a walk has to visit the same vertex twice. We thus get that such a walk exists if and only if $A$ has a directed cycle. The algorithm, therefore, decides the class of digraphs that contain a directed cycle.
\end{example}

\begin{example}
\label{example:decide-every-struct-lqa}
    In this example, we show that every class $\C$ of structures is decided by an adaptive left unbounded query algorithm over $\N$.
    
    The algorithm first queries for the size of the structure. Then it queries the number of homomorphisms from every structure that is not larger than it. 
    By the proof of Lovász's theorem (see, for example, Section III of \cite{atserias2021expressive} for a proof of this), the algorithm has now distinguished the structure (up to isomorphism) and can thus classify it correctly. 
    \par 
    To formally define the algorithm, let $\tau$ be the signature of the structures in $\C$. We fix an enumeration $A_1, A_2, A_3,\ldots$ of all structures of signature $\tau$ such that for each $n \in \N$ there exists a number $r_n$ such that $A_1, \ldots, A_{r_n}$ is an enumeration of all of the structures of size at most $n$.
    Then the algorithm can be defined with: \begin{align*}
        G(\sigma) \coloneqq \begin{cases}
            (\{0\}, \varnothing) &\If \sigma = \varepsilon \\
            A_{k+1} &\If \sigma = `na_1 a_2 \ldots a_k\textrm' \text{ and } k<r_n\\
            \textsf{YES} &\If \sigma =  `na_1 a_2 \ldots a_{r_n}\textrm' \text{ and the unique $A$ s.t. $\hom_\N(A_i, A) = a_i$}\\
            &\hspace{17pt}\text{for each $1 \leq i \leq r_n$ satisfies $A \in \C$}\\ 
            \textsf{NO} &\If \sigma =  `na_1 a_2 \ldots a_{r_n}\textrm' \text{ and the unique $A$ s.t. $\hom_\N(A_i, A) = a_i$}\\
            &\hspace{17pt}\text{for each $1 \leq i \leq r_n$ satisfies $A \notin \C$} 
        \end{cases}
    \end{align*}
\end{example}

\section{Query Algorithms over $\N$}
\label{left-algs-N}

\subsection{Unboundedness helps for adaptive left query algorithms}
\label{sec:unb helps}
A starting point for our work in this section is the following result by Chen et al:
\begin{theorem}[Theorem 8.2 in \cite{chen_et_al:LIPIcs.MFCS.2022.32}]
    Every class of simple graphs can be decided by an adaptive left 3-query algorithm over $\N$.
\end{theorem}

This result invites the question of what happens in the general case, for arbitrary relational structures.
We prove that the same does not hold in this broader setting. 
We even show that there is a class of directed graphs that is not decided by an adaptive $k$-query algorithm over $\N$, for any $k$.
To do this, we utilize the symmetry of directed cycles to prove that they are hard to distinguish from each other using left homomorphism queries.
\par 
We begin with a lemma that shows that the homomorphism count from a structure to a disjoint union of cycles can be described simply.

\begin{restatable}{lemma}{homLemma}
\label{hom-lemma}
For every digraph $A$ and all positive integers $n,m$ we have
    \begin{align*}
        \hom(A, m\cdot \mathsf{C}_n) \coloneqq \begin{cases}
            0 &\If n \nmid \gamma(A)  \\
            (m\cdot n)^{c(A)} &\If n \mid \gamma(A)
        \end{cases}.
    \end{align*}
\end{restatable}

\par 
The proof of this theorem is in Appendix \ref{appendix:basic-digraphs}. 
The proof of this theorem can be intuitively explained as follows:
In disjoint unions of directed cycles, each vertex has in-degree 1 and out-degree 1. 
For a homomorphism from $A\to m \cdot \mathsf{C}_n$, the whole map is therefore decided by the value of the map in one vertex from each component of $A$.
Since there are $m\cdot n$ ways to pick the value of a homomorphism on that initial vertex (if such a homomorphism exists), the lemma follows.
\smallskip \par 

We say that two classes $\A,\mathcal B$ of structures are separated by a query algorithm if it accepts all structures in $\A$ and rejects all structures in $\mathcal B$, or vice versa.
\par 
The importance of the above lemma is in the fact that it shows that for a given structure $m \cdot \mathsf{C}_n$ there are only two possible outcomes for $\hom(F,  m \cdot \mathsf{C}_n)$. This limits what can be done using few queries, as is shown in the following theorem:

\begin{theorem}
\label{main-theorem}
    The class $\D_n \coloneqq\Big\{2^{n-m} \cdot \mathsf{C}_{2^{m}} : 0\leq m \leq n  \hspace{6pt}\text{and}\hspace{6pt} m \text{ is even }\Big\}$ can be separated from the class 
    $\D_n'\coloneqq\Big\{2^{n-m} \cdot \mathsf{C}_{2^m} : 0\leq m \leq n  \hspace{6pt}\text{and}\hspace{6pt} m \text{ is odd }\Big\}$ by a non-adaptive left $k$-query algorithm over $\N$ if and only if $k\geq n$.
\end{theorem}
\proof By Lemma \ref{hom-lemma} we have that for every digraph $F$ and every $2^{n-m} \cdot \mathsf{C}_{2^{m}} \in \D_n \cup \D_n'$:\begin{align*}
    \hom(F,2^{n-m} \cdot \mathsf{C}_{2^{m}}) 
    &= \begin{cases}
        0 &\If 2^m \nmid \gamma(F) \\
        (2^n)^{c(F)} &\If 2^m \mid \gamma(F)
    \end{cases} \\
        &=\begin{cases}
        0 &\If  \nu_2(\gamma(F)) <m \\
        (2^n)^{c(F)} &\If \nu_2(\gamma(F)) \geq m
    \end{cases}
\end{align*}
where $\nu_2(k)$ denotes the largest positive integer $r$ such that $2^r \mid k$ if $k\neq 0$, if $k = 0$ then $\nu_2(k) = + \infty$.
It follows from this that to have $\hom(F,2^{n-m} \cdot \mathsf{C}_{2^{m}}) \neq \hom(F,2^{n-(m+1)} \cdot \mathsf{C}_{2^{m+1}})$ we must have $\nu_2(\gamma(F)) = m+1$. To distinguish between $2^{n-m} \cdot \mathsf{C}_{2^{m}}$ and $2^{n-(m+1)} \cdot \mathsf{C}_{2^{m+1}}$ for every $m \in \{0, \ldots, n-1\}$ the algorithm then needs to query some $F_m$ with $\nu_2(\gamma(F_m)) = m$ for each $m \in \{0,\ldots, n-1\}$. The algorithm, therefore, needs at least $n$ queries to separate the classes. 
It is clear that $n$ non-adaptive queries suffice: the algorithm can query $\mathsf{C}_{2^m}$ for each $m \in \{0,\ldots, n-1\}$. Every structure in $\D_n \cup \D_n'$ yields a unique outcome from the collection of these queries, and thus they can be classified correctly. \par
It is worth noting that homomorphism existence queries suffice here, and hence
$\D_n$ and $\D'_n$ are already separated by a non-adaptive left $n$-query algorithm over $\B$.
\qed 

\begin{corollary}
    \label{intermediate-corollary}
    The class $\D_n$ can be separated from the class 
    $\D_n'$ by an adaptive left $k$-query algorithm over $\N$ if and only if $k\geq\log (n+1)$.

    \end{corollary}
\proof 
For any digraph $H \in \D_n\cup \D_n'$, there are only two possible outcomes for every query $\hom(F, H)$ (as explained above). Therefore, an adaptive left $k$-query algorithm over $\N$ that separates $\D_n$ from $\D_n'$ queries at most 
\[2^0 + 2^1 + \ldots + 2^{k-1} = 2^k -1\]
different structures in all of its possible computation paths on inputs from $\D_n\cup \D_n'$. It can thus be translated into a non-adaptive left $(2^k-1)$-query algorithm over $\N$ separating the two classes. By Theorem \ref{main-theorem} it follows that such a non-adaptive left $(2^k-1)$-query algorithm over $\N$ exists if and only if $2^k-1 \geq n$, so $k\geq \log (n+1)$. Thus, we have shown that the classes $\D_n$, $\D_n'$ can only be separated by an adaptive left $k$-query algorithm over $\N$ if $k\geq \log(n+1)$.
It is also clear that $k\geq \log(n+1)$ suffices. The algorithm can use queries of the form $\hom(\mathsf{C}_{2^{r}}, 2^{n-m} \cdot \mathsf{C}_{2^m})$ to binary search for the $m$ of an input structure of the form $2^{n-m} \cdot \mathsf{C}_{2^m}$, and use that to classify it correctly.
\qed

Let $\mathfrak N_k$ denote the set of classes that are decided by a non-adaptive left $k$-query algorithm over $\N$, and $\mathfrak A_k$ denote the set of classes that are decided by an adaptive left $k$-query algorithm over $\N$.
Theorem \ref{main-theorem} and Corollary \ref{intermediate-corollary} show, among other things, that \begin{align*}
    \mathfrak N_1 \subsetneq \mathfrak N_2 \subsetneq \mathfrak N_3 \subsetneq\ldots
    \quad\text{and}\quad
     \mathfrak A_1 \subsetneq \mathfrak A_2 \subsetneq \mathfrak A_3 \subsetneq \ldots
\end{align*}

Using Corollary \ref{intermediate-corollary}, we can also prove an even stronger result:
\begin{theorem}
    \label{better-corollary}
    Let $\C$ be the class of digraphs whose smallest directed cycle (if it exists) has length that is an even power of two. Then $\C$ is not decided by an adaptive left $k$-query algorithm over $\N$ for any $k$.
\end{theorem}
\proof Note that we have 
\[\bigcup_{n>0} \D_n \subseteq \C \subseteq (\bigcup_{n>0} \D_n')^c,\]
where $X^c$ denotes the set theoretic complement of a set $X$.
A given adaptive left $k$-query algorithm over $\N$ cannot separate $\D_{2^k}$ from $\D_{2^k}'$, by Corollary \ref{intermediate-corollary}. To decide $\C$ it must in particular separate these classes, therefore the algorithm can not decide $\C$. Thus it is clear that $\C$ is not decided by any adaptive left $k$-query algorithm over $\N$ for any $k$.\qed

\smallskip\par 
Using the same method it can be shown that many other classes are not decided by an adaptive left $k$-query algorithm over $\N$ for any $k$. In fact, to show that a class is not decided by such an algorithm it suffices to show that it separates $\D_n$ and $\D_n'$ for infinitely many $n$.
Another example of such a class is the class of digraphs whose number of components is an even power of two.
\begin{note*}
    As we saw in Example \ref{example:decide-every-struct-lqa}, every class is decided by an adaptive left unbounded query algorithm over $\N$.
    Theorem \ref{better-corollary} therefore also shows that unboundedness increases the expressive power of adaptive left query algorithms over $\N$, i.e. $\bigcup_{k} \mathfrak A_k \subsetneq \mathfrak{A}_{unb}$ where $\mathfrak{A}_{unb}$ denotes the class of structures decided by an adaptive left unbounded query algorithm over $\N$. \looseness=-1
\end{note*}

The preceding results show that there exist classes of digraphs that adaptive left bounded query algorithms over $\N$ do not decide. It follows immediately that for every signature containing a binary relation, there exists a class of structures of that signature that left bounded query algorithms over $\N$ cannot decide. For signatures containing only unary relations, it can be shown that every class can be decided with an adaptive left bounded query algorithm over $\N$
(see Theorem \ref{thm:unary-is-decidable}).
What about signatures that contain some $n$-ary relation with $n>2$ and no binary relation? 
The following theorem answers this question.
\begin{restatable}{theorem}{nAryBetterCorollary}
    \label{n-ary-better-corollary}
    For every $n\geq 2$ and every signature $\tau$ containing an $n$-ary relation, there exists a class of structures with signature $\tau$ that is not decided by an adaptive left $k$-query algorithm over $\N$ for any $k$.
\end{restatable}
The class we construct in the proof of this theorem is similar to the class from Theorem \ref{better-corollary}, except we replace the directed cycles with $n$-ary ``cycle-like'' structures.

\subsection{Adaptive versus non-adaptive left $k$-query algorithms}
\label{sec:adaptive vs non-adaptive}
\smallskip\par 
We now turn to a comparison of the expressive power of adaptive and non-adaptive left query algorithms over $\N$.
It follows from Theorem \ref{main-theorem} and Corollary \ref{intermediate-corollary} that $\mathfrak{N}_{2^k} \nsubseteq \mathfrak{A}_k$. We also trivially have the relation $\mathfrak N_k \subseteq \mathfrak A_k$. This raises the question of where the bound lies, that is, for which $l$ we have $\mathfrak N_l \subseteq \mathfrak A_k$ but $\mathfrak N_{l+1} \nsubseteq \mathfrak A_k$? This is answered by the following theorem:

\begin{theorem}
\label{adaptive-not-better}
    For each $k$, there exists a class of structures that is decided by a non-adaptive left $k$-query algorithm over $\N$ but not by an adaptive left $(k-1)$-query algorithm over $\N$.
\end{theorem}
\proof 
Let $p_1, \ldots, p_{2k}$ be distinct primes and write $P \coloneqq \prod_{i = 1}^{2k} p_i$ and $q_i \coloneqq P/p_i$.
Define a non-adaptive left $k$-query algorithm $((F_1, \ldots, F_k), X)$ with \begin{align*}
    F_i \coloneqq p_i \cdot \mathsf{C}_{q_i}
    \hspace{6 pt} \text{and} \hspace{6 pt}
    X \coloneqq \{ (P\cdot \delta_{1,j}, \ldots, P\cdot \delta_{k,j}) : j \in \{1,\ldots, k\}\}
\end{align*}
where $\delta_{i,j} = \begin{cases}
    1 \If i =j \\
    0 \If i\neq j
\end{cases}$ is the Kronecker delta.
Let $\C$ be the class decided by the above query algorithm. 
It follows from Lemma \ref{hom-lemma} that \begin{align*}
    \hom(F_i , p_j \cdot \mathsf{C}_{q_j}) = P \cdot \delta_{i,j},
\end{align*}
so for $j \in \{1,\ldots, 2k\}$ we have $p_j \cdot \mathsf{C}_{q_j} \in \C$ if and only if $j\leq k$. We will show that $\C$ cannot be decided by an adaptive left $(k-1)$-query algorithm over $\N$. 

Let $G$ be a given adaptive left $(k-1)$-query algorithm over $\N$. 
Using an adversarial argument, we find two structures that are classified in the same way by $G$ but differently by $\C$. We do this by defining sets $\A_n$ such that $G$ has the same computation path on all elements of $\A_n$ up to stage $n$.
We define $\A_0 \coloneqq \{ p_j \cdot \mathsf{C}_{q_j} : j \in \{1,\ldots, 2k\}\}$. 
Now let $\A_n\subseteq \A_0$ be given such that $G$ has the same computation path $\sigma$ on all elements of $\A_n$ up to stage $n$.
Write $F \coloneqq G(\sigma)$.
We have two cases: \begin{itemize}
    \item If $\hom(F, A)$ is the same for all elements $A$ of $\A_n$ we can simply define $\A_{n+1} \coloneqq \A_n$ and the computation path is the same up to stage $n+1$.
    \item Assume otherwise. 
    Note that for any element $p_j \cdot \mathsf{C}_{q_j}$ of $\A_0$ we have by Lemma \ref{hom-lemma} that \begin{align*}
        \hom(F,p_j \cdot \mathsf{C}_{q_j}) = \begin{cases}
            0 &\If q_j \nmid \gamma(F) \\
            P^{c(F)} &\If q_j \mid \gamma(F).
        \end{cases}
    \end{align*}
    So there are only two possible outcomes. Moreover, if $q_i, q_j \mid \gamma(F)$ for $i \neq j$, then $P \mid \gamma(F)$ and thus $q_r \mid \gamma(F)$ for all $r \in \{1, \ldots, 2k\}$ and we are in the former case. Thus, there is an $i$ such that $q_j \mid \gamma(F)$ if and only if $j = i$. We can thus define $\A_{n+1} \coloneqq \A_n \setminus\{p_i \cdot \mathsf{C}_{q_i}\}$. It is clear that $\hom(F, A) = 0$ for all elements $A$ of $\A_{n+1}$, so they share the same computation path up to stage $n+1$.
\end{itemize}
The above construction gives us a set $\A_{k-1}$ such that $G$ has the same computation path up to stage $k-1$ on all elements of it. Clearly we also have that $\abs{\A_{k-1}} \geq \abs{\A_0}- (k-1) = k+1$. This means that there must exist $p_i \cdot \mathsf{C}_{q_i}, p_j \cdot \mathsf{C}_{q_j} \in \A_{k-1}$ such that $i \leq k$ and $j\geq k+1$. Thus, we have found elements that $G$ classifies in the same way but $\C$ classifies differently. We have thus shown that $G$ does not decide $\C$.
\qed
\smallskip \par 

The above theorem describes a class of structures where adaptiveness does not help at all when trying to decide the class; therefore, we have that $\mathfrak N_{k+1} \nsubseteq \mathfrak A_k$ for each $k$. This is interesting as in many cases, adaptiveness greatly reduces the number of queries needed.
\par 
In this regard, we now prove that there is an adaptive left 2-query algorithm over $\N$ that decides a class that no non-adaptive left query algorithm over $\N$ decides. 
\begin{theorem}
    \label{A2 not non-adaptive}
    The class of all digraphs that contain a directed cycle is decided by an adaptive left 2-query algorithm over $\N$, but not by any non-adaptive left $k$-query algorithm over $\N$, for any $k$.
\end{theorem}
\proof Note that directed cycles are preserved by homomorphisms. It then follows that the class is closed under homomorphic equivalence. It is also well known that this class is not first-order definable; this is a standard application of the Ehrenfeucht-Fra\"{i}sse method, which can be found in textbooks on finite model theory \cite{Ebbinghaus1995}.
It then follows from Theorem \ref{Balder main theorem} that the class is not decided by a non-adaptive left $k$-query algorithm over $\N$, for any $k$.
\par 
In example \ref{example:2-adaptive-to-decide-directed-cycle} we saw an adaptive left 2-query algorithm over $\N$ that decides the class. This completes the proof.
\qed
\smallskip\par 
    The preceding results now give the relations between $\mathfrak N_k$ and $\mathfrak A_k$ shown in Figure~\ref{fig:summary-left-N}.

\subsection{Adaptive right query algorithms over $\N$}
\label{right-algs-over-N}

In their paper, Chen et al. also prove a result about adaptive \emph{right} query algorithms on simple graphs:

\begin{proposition}[Proposition 9.3 in \cite{chen_et_al:LIPIcs.MFCS.2022.32}]
    There exists a class of simple graphs that is not decided by an adaptive right $k$-query algorithm over $\N$ for any $k$, using only simple graphs in its queries.
\end{proposition}

In the general case, we again get a different result. It turns out that using looped structures in the queries changes the picture completely. 
This result was essentially proved by Wu \cite{wu2023study}, but it is not stated in this form in his work.

\begin{restatable}[From \cite{wu2023study}]{theorem}{rightQueryTheorem}
\label{right-query-theorem}
    Every class of structures can be decided by an adaptive right 2-query algorithm over $\N$.
\end{restatable}
This result also is in contrast to our result for left query algorithms:
the hierarchy of expressive power of adaptive \emph{right} $k$-query algorithms collapses at $k=2$, in contrast to the strict hierarchy shown in Figure \ref{fig:summary-left-N} for adaptive \emph{left} $k$-query algorithms.
The cause of this is that for right query algorithms, we have more constructions that give useful information about the homomorphism counts. In particular, we have that if $H$ is connected then \begin{align*}
    \hom(H, A \oplus B) = \hom(H, A) + \hom(H, B).
\end{align*} No construction on the left side allows for addition of the hom-vectors.

\section{Query Algorithms Over $\B$}
\label{algs-over-B}

We now move on to studying adaptive query algorithms over $\B$. 
We focus on left query algorithms, and we will comment on the 
case of right query algorithms afterwards.

\subsection{Unboundedness helps for adaptive left query algorithms}
\label{sec:unb helps B}

Every class that can be decided by an adaptive $k$-query algorithm over $\B$ can also be decided by a non-adaptive $(2^{k}-1)$-query algorithm over $\B$. This is done by querying all the structures that the algorithm can possibly query in the first $k$ steps of its run. Since there are at most $2$ branches from each node in the computation tree, this is at most $1+ 2+ 2^2+ \ldots+2^{k-1}\leq (2^{k}-1)$ queries. This means that \emph{adaptive} and \emph{non-adaptive} left bounded query algorithms over $\B$ have the same expressive power.

The question we answer below is whether every class decided by an adaptive unbounded query algorithm over $\B$ can be decided by an adaptive bounded query algorithm (and thus a non-adaptive one)? The answer is no.

\begin{theorem}
\label{leftUnboundedMoreExpressive}
    The class $\C = \{A : A \text{ is a digraph that contains a directed cycle}\}$ can be decided by an adaptive left unbounded query algorithm over $\B$, but not by an adaptive left $k$-query algorithm over $\B$, for any $k$.
\end{theorem}
\proof In theorem \ref{A2 not non-adaptive} we showed that $\C$ is not decided by a non-adaptive left $k$-query algorithm over $\N$, it is thus not decided by a non-adaptive left $k$-query algorithm over $\B$.
\par 
We now provide an adaptive left unbounded query algorithm over $\B$ that decides $\C$.
The algorithm runs as follows: \par 
Let $A$ be the given input.
For $n = 1,2, \ldots$ query $\hom_\B(\mathsf{P}_n, A)$ and $\hom_\B(\mathsf{C}_n, A)$. If at any point $\hom_\B(\mathsf{P}_n, A) = 0$ then halt and output \textsf{NO}. If at any point $\hom_\B(\mathsf{C}_n, A) = 1$ then halt and output \textsf{YES}.
\smallskip \par 
To see that this algorithm decides $\C$ we note that the following equivalence holds: \begin{align*}
    A \text{ has no directed cycle }\iff 
    &\text{the lengths of directed walks in $A$ are bounded} \\&\text{by a constant $k$} \\
    \iff &\hom_\B(\mathsf{P}_n, A) = 0 \text{ for some } n.
\end{align*} Also, directed cycles are preserved under homomorphism, so $A$ has a directed cycle if and only if $\hom_\B(\mathsf{C}_n, A) =1$ for some $n$. Thus, it is clear that the algorithm halts on all inputs and produces the correct answer.
\qed

\subsection{Characterization of the expressive power of adaptive left unbounded query algorithms over $\B$}
\label{Unbounded qals over B}
We can use the language of topology to characterize the expressive power of adaptive unbounded query algorithms exactly. 
For background on topology, we refer the reader to \cite{Munkres2000Topology}.
The topology we define is generated by all ``basic observations'', where a basic observation is a class defined by the answer to a single homomorphism query.
\par 
More formally, for a structure $A \in \mathsf{FIN}(\tau)$ we define \[\up A \coloneqq \{B : A \to B\},\] which, phrased in terms of 
\emph{conjunctive queries} (CQs) is simply the set of structures that
satisfy the canonical CQ $q_A$ of $A$. 
We then let 
\begin{align*}
    \mathcal S \coloneqq \{ \up A  : A \in \mathsf{FIN}(\tau) \} \cup \{ \mathsf{FIN}(\tau) \setminus \up A : A \in \mathsf{FIN}(\tau)\}
\end{align*}
be the subbasis for our topology $\mathfrak T$. 
In other words, phrased again in terms of conjunctive queries, the open sets of the topology
are all classes of structures that can be defined by a CQ or the negation of a CQ, as well as 
all classes generated from these by closing under finite intersection and arbitrary union.
\par 

We have the following characterization of classes that are decided by adaptive left unbounded query algorithms over $\B$.

\begin{restatable}{theorem}{topologicalCharacterization}
\label{Topological Characterization}
    A class $\C \subseteq \mathsf{FIN}(\tau)$ is decided by an adaptive left unbounded query algorithm over $\B$ if and only if $\C$ is a clopen set in the topological space $(\mathsf{FIN}(\tau), \mathfrak T)$.
\end{restatable}
Remember, a \emph{clopen} set is a set that is both closed and open.
We say that a class $\C$ of structures is homomorphism-closed if whenever $A \in \C$ and $A \to B$ for some structure $B$, then $B \in \C$.
We can use the topological characterization to get a simpler characterization for the existence of an adaptive left unbounded query algorithm over $\B$ in the case of homomorphism-closed classes.

\begin{restatable}{theorem}{datalogUnboundedLqaThm}
\label{datalogUnboundedLqaThm}
    Let $\C$ be a homomorphism-closed class; then the following hold:
    \begin{enumerate}[(i)]
        \item $\C = \bigcup_{A \in \C} \up{A}$, hence $\C$ is open.
        \item Furthermore, $\C$ is closed if and only if $\C = \bigcap_{i \in I} \bigcup_{j = 0}^{n_i}\up{A_{i,j}}$ for some structures $A_{i,j}$. 
    \end{enumerate}
\end{restatable}
\smallskip \par
The preceding two theorems imply that a homomorphism-closed class $\C$ is decided by an adaptive left unbounded query algorithm over $\B$ if and only if $\C$ is an intersection of UCQs.
\begin{remark}
    The first part of Theorem \ref{leftUnboundedMoreExpressive} can be derived from Theorem \ref{datalogUnboundedLqaThm} since $\C$ is homomorphism-closed and $\C = \bigcap_{n\geq1} \mathsf{P}_n$.
\end{remark}

\subsection{Case Studies}
\label{sec:case studies}
As our first case study, we consider the definability of \emph{Constraint Satisfaction Problems} (CSPs) with adaptive left unbounded query algorithms over $\B$.
\par 
A CSP can be formulated as a problem of deciding whether there is a homomorphism from an input structure to a target structure. For a relational structure $A$ we write $\CSP(A) \coloneqq \{ B : B \to A\}$.
\par 
We can combine some known results with our observations from Section \ref{Unbounded qals over B} to get the following result:
\begin{restatable}{theorem}{mainCSP}
\label{thm:main-CSP}
Let $\C$ be any class of the form $[A]_{\leftrightarrow}$ or $\CSP(A)$ for a structure $A$. Then the following are equivalent:
    \begin{enumerate}[(i)]
        \item $\C$ is decided by an adaptive left unbounded query algorithm over $\B$.
        \item $\C$ is decided by a non-adaptive left $k$-query algorithm over $\B$ for some $k$.
    \end{enumerate}
\end{restatable}
\par 
This theorem shows that unboundedness does not help when deciding CSPs or homo\-morphic equivalence classes, and it gives an effective characterization of the classes of this form that are decided by an adaptive left unbounded query algorithm over $\B$:
\par 
\begin{corollary}
\label{cor:csp-FO}
    $\CSP(A)$ is decided by an adaptive left unbounded query algorithm over $\B$ if and only if it first-order definable. Moreover, there is an effective algorithm that decides for a given structure A whether this holds. 
\end{corollary}

The fact that $\CSP(A)$ is decided by a non-adaptive left query algorithm over $\B$ if and only if it is first-order definable follows from Theorem \ref{Balder main theorem}.
The fact that there is an algorithm that decides whether $\CSP(A)$ is first-order definable is a well-known result by Larose, Loten, and Tardif \cite{larose2007characterisation}. 
\par 
Consider now the database query language Datalog. To conserve space, we omit a detailed definition of the syntax and semantics of Datalog, which can be found in \cite{Ebbinghaus1995}.
Here we only consider Boolean Datalog programs, which are Datalog programs with a 0-ary goal predicate. Such a program $\pi$ defines a class of structures $\C_\pi$ that is homomorphism-closed. A Datalog program is \emph{monadic} if it only uses recursive predicates (a.k.a.~IDBs) that are unary; it is \emph{linear} if every rule body has at most one occurrence of a recursive predicate.

The class $\C$ of digraphs that contain a directed cycle from Theorem \ref{leftUnboundedMoreExpressive} is Datalog definable, using the program: 
\[
    \begin{array}{lll}
        X(x,y) &\leftarrow & R(x,y) \\
        X(x,y) & \leftarrow & X(x, x_1), R(x_1,y)\\
        Ans() & \leftarrow & X(z,z)
    \end{array}
\]
This shows that unboundedness helps even when we restrict ourselves to Datalog definable classes. 
Contrastingly, Theorem \ref{thm:main-CSP} shows that for every Datalog program that defines the complement of a CSP, unboundedness does \emph{not} help. 
Corollary \ref{cor:csp-FO} furthermore shows that a class defined by such Datalog programs is decided by an adaptive left unbounded query algorithm if and only if it is first-order definable.
Note that the Datalog programs that define the complement of a CSP form a non-trivial fragment of Datalog. 
Indeed, every (Boolean) monadic Datalog program whose rule bodies are ``tree-shaped'' defines a complement of a CSP (and moreover, the CSP in question can be constructed effectively from the Datalog program),  cf.~\cite{Erdos2017:regular,tencate2024:right}. 

We can use these observations to prove a negative result on the expressive power of adaptive left query algorithms over $\B$. 

\begin{restatable}{theorem}{monadicDatalog}
\label{thm:monadicDatalog}
    There exists a class definable by a monadic linear Datalog program that is not decided by an adaptive left unbounded query algorithm over $\B$.
\end{restatable}
\par 
We now give a characterization of the Datalog definable classes that are also decided by an adaptive left unbounded query algorithm in terms of their \emph{upper envelopes}.
\par
Chaudhuri and Kolaitis \cite{chaudhuri1993finding,chaudhuri1994can}, 
in their study of approximations of Datalog programs by non-recursive queries, they introduced the notion of an 
upper envelope. 
For a Boolean Datalog program deciding a class $\C$, an
upper envelope is a union of conjunctive
queries (UCQ) $q$ that defines a class $\C'$ such 
that $\C\subseteq \C'$. Such an upper envelope can be equivalently viewed as a class $\C'$ of
the form $\bigcup_{j = 0}^{n}\up{A_{j}}$ that contains $\C$. Since Datalog definable classes $\C$ are homomorphism-closed,
Theorem \ref{datalogUnboundedLqaThm} therefore gives us the following:
\begin{corollary}
\label{cor:datalog-char}
    A Datalog definable class $\C$ is decided by an adaptive left unbounded query algorithm over $\B$ if and only if it is the intersection of its upper envelopes.
\end{corollary}

\subsection{Adaptive Right Query Algorithms over $\B$}
\label{sec:rqa over B}

For adaptive \emph{right} query algorithms over $\B$ we have a result analogous to Theorem \ref{leftUnboundedMoreExpressive}:
\begin{restatable}{theorem}{unboundedHelpsRqa}
\label{unbounded-helps-rqa}
    The class $\C$ of all digraphs that have an oriented cycle with non-zero net length can be decided by an adaptive right unbounded query algorithm over $\B$, and not by an adaptive right $k$-query algorithm over $\B$, for any $k$.
\end{restatable}
\par 
It can be shown that this class is Datalog definable
(see Theorem \ref{thm:non-zero-net-length-is-Datalog} in Appendix \ref{appendix:rqa-thm}).
So, as in the case for left query algorithms, unboundedness sometimes helps even when we restrict ourselves to Datalog definable classes.

\smallskip \par 
We can also replicate the results from section \ref{Unbounded qals over B} for right query algorithms.
We define $\mathfrak T'$ as the topology generated by the sub-basis \begin{align*}
    \mathcal S' \coloneqq \{ \down{A} : A \in \mathsf{FIN}(\tau) \} \cup \{ \mathsf{FIN}(\tau) \setminus \down {A} : A \in \mathsf{FIN}(\tau)\}.
\end{align*} 
where $\down{A}$ denotes the set $\{ B \in \mathsf{FIN}(\tau) : B\to A \}$. 
In other words, $\mathcal S'$ consists of all CSPs and their complements.
Using the same argument as before, we can prove:
\begin{theorem}
    \label{topological right characterization}
     A class $\C \subseteq \mathsf{FIN}(\tau)$ is decided by an adaptive right unbounded query algorithm over $\B$ if and only if $\C$ is a clopen set in the topological space $(\mathsf{FIN}(\tau), \mathfrak T')$.
\end{theorem}

Then we also get 
\begin{theorem} 
\label{thm:rqa-unb-does-not-help}
Let $\C$ be any class of the form $[A]_{\leftrightarrow}$ or $\up{A}$ for a structure $A$. Then the following are equivalent:
    \begin{enumerate}[(i)]
        \item $\C$ is decided by an adaptive right unbounded query algorithm over $\B$.
        \item $\C$ is decided by a non-adaptive right $k$-query algorithm over $\B$ for some $k$.
    \end{enumerate}
\end{theorem}

\par

It then follows from a result by ten Cate et al. \cite{ten2024homomorphism} that these conditions hold if and only if $A$ is homomorphically equivalent to an acyclic structure.
Using the language of conjunctive queries, this shows that a CQ $q$ has an equivalent adaptive right unbounded query algorithm if and only if $q$ is equivalent to a Berge-acyclic CQ. 
It is also good to note that since every CQ can be written as a Datalog program, this also shows that there are Datalog-definable classes that are not decided by any adaptive right unbounded query algorithm over $\B$.

\section{Open Problems}
\label{sec:discussion}

\subparagraph*{Adaptive $f(n)$-query algorithms.} 
We can define an adaptive $f(n)$-query algorithm as an adaptive unbounded query algorithm that uses at most $f(n)$ queries for inputs of size $n$. Then one could ask the question of what the growth rate of $f$ needs to be for adaptive left $f(n)$-query algorithms over $\N$ to be able to decide any class of structures of a given signature.
The algorithm described in Example \ref{example:decide-every-struct-lqa} gives an upper bound of $2^{\Sigma_i(n^{r_i})}$ for such an $f$, where $r_i$ are the arities of the relation symbols in the signature.
On the other hand, Corollary \ref{intermediate-corollary} implies that such an $f$ is not $o(\log\log n)$.
This leaves quite a large gap between the known lower and upper bounds. Shrinking this gap is an interesting avenue for further research.

\subparagraph*{Decidability of the existence of an adaptive left unbounded query algorithm over $\B$.}
For classes of the form $\CSP(A)$, Corollary \ref{cor:csp-FO} shows that determining whether they are decided by an adaptive left unbounded query algorithm is decidable. The criterion given in Corollary \ref{cor:datalog-char} for Datalog definable classes is, however, not effective. This raises the question of whether there exists such an effective criterion for Datalog definable classes, or whether that problem is undecidable.
Note that the ``Rice-style'' general undecidability result for semantic properties of Datalog programs from \cite{Gaifman93:undecidable} appears to be of no help here: although \emph{admitting an adaptive left unbounded query algorithm over $\B$} is a semantic property of Datalog programs, it fails to satisfy their ``strong non-triviality'' requirement.

\subparagraph*{Adaptive hybrid $k$-query algorithms over $\B$.}
A query algorithm that can use both left and right queries can be called a hybrid query algorithm. 
Adaptive hybrid query algorithms over $\N$ have been studied by Wu \cite{wu2022query}, but there has been no research on such algorithms over $\B$.
Every class of structures that is closed under homomorphic equivalence can be decided by an adaptive hybrid unbounded query algorithm over $\B$. On input $B$, the algorithm simply queries for $\hom_{\B}(A, B)$ and $\hom_\B(B,A)$ for each structure $A$. At some point, both of these numbers are 1, then the homomorphic equivalence class of $B$ is decided, which suffices to classify it correctly. 
This evokes the question of what the expressive power of hybrid $k$-query algorithms is: Do we get a strict hierarchy as in Figure \ref{fig:summary-left-N}? What is the relation between the expressive power of non-adaptive and adaptive hybrid query algorithms over $\B$?

\bibliography{MFCS2025}

\appendix
\section{Proofs of Statements About Digraphs and Homomorphisms}
\label{appendix:basic-digraphs}
We first define a few basic concepts:

\begin{definition}\hspace{0cm}
    \begin{itemize}
        \item Let $W = (a_0, r_0, a_1, r_1, \ldots, a_n, r_{n}, a_{n+1})$ be an oriented walk. The inverse of $W$ is  \begin{align*}
        -W \coloneqq (a_{n+1}, r_n' , a_n , \ldots, r_1', a_1, r_0', a_0)
        \end{align*} where $r_i'$ is the opposite of $r_i$, so $r_i' \coloneqq  \begin{cases}
            + \If r_i = -\\
            - \If r_i = +
        \end{cases}$.

        \item For oriented walks $W = (a_0, r_0, \ldots, a_{n+1})$ and $W' = (b_0, s_0, \ldots, b_{n+1})$ with $a_{n+1} = b_0$ we can define the composition of $W$ and $W'$ as \begin{align*}
        W \circ W' = (a_0, r_0, \ldots, a_{n+1} = b_0 , s_0, \ldots, b_{n+1}).
    \end{align*}
    \end{itemize}
\end{definition}

We can now prove:
\girthProp*
\proof 
The left-to-right direction follows immediately from Lemma \ref{girth-lemma}, since every cycle of positive net length is a closed oriented walk. 

Let us now prove the other direction. We assume $n \mid \gamma(A)$. 
By Lemma \ref{girth-lemma} it suffices to show $\gamma(A) \mid l(W)$ for every closed oriented walk $W= (a_0, r_0, a_1, r_1, \ldots, a_{n-1}, r_{n-1}, a_{n})$ in $A$. To do that we use strong induction over $n$, the length of the walk (not the net length). We have two cases:
\begin{itemize}
    \item Assume $W$ is an oriented cycle. If $l(W) = 0$ then trivially $\gamma(A) \mid l(W)$. If $l(W)>0$ then by definition $\gamma(A) \mid l(W)$. Lastly, if $l(W)<0$ then $l(-W)>0$ so $-W$ has positive net length. Then $\gamma(A) \mid l(-W)$ but $l(-W) = -l(W)$ so $\gamma(A) \mid l(W)$.
    
    \item Assume $W$ is not an oriented cycle. Then there exists a subwalk $C = (a_i, r_i, \ldots, a_{j+1})$ (with $i\leq j$) that is an oriented cycle. But then \begin{align*}
        W' = (a_0, r_0, \ldots, a_{i-1}, r_{i-1}, a_i = a_{j+1}, r_{j+1}, \ldots, a_{n-1}, r_n, a_{n})
    \end{align*} is a smaller closed oriented walk (it has length $n-(j-i+1)<n$). So by the induction hypothesis, we get that $\gamma(A) \mid l(W')$. We also have $l(W) = l(W') + l(C)$ and by the argument from the previous case, $\gamma(A) \mid l(C)$. So we have $\gamma(A) \mid l(W)$.
\end{itemize}
Using induction, we conclude that $\gamma(A) \mid l(W)$ for all closed oriented walks $W$ in $A$, which is what we wanted to show.
\qed

We are now going to prove the lemma that is the cornerstone for our limitative result on the expressive power of left query algorithms. Before we do that we need to establish some basic facts.
\par 

For a given signature $\tau = \{R_1, \ldots, R_n\}$, we let $\mathbf{1}_\tau$ denote the complete singleton structure of that signature, namely \begin{align*}
    \mathbf{1}_\tau \coloneqq (\{0\}, (\{0\}^{r_i})_{R_i \in \tau}).
\end{align*}
For digraphs, this is simply the structure consisting of a single loop. If the signature is clear from the context, we usually drop the subscript. 
\par 
The \emph{direct product} of two structures $\mathfrak A = (A, (R_i^\mathfrak A)_{i\in \{1,\ldots,n\}})$ and $\mathfrak B = (B, (R_i^\mathfrak B)_{i\in \{1,\ldots,n\}})$ is defined with
\[\mathfrak A \otimes \mathfrak B \coloneqq (A \times B, R_i^{\mathfrak A \otimes \mathfrak B})\]
where 
\[R_i^{\mathfrak A \otimes \mathfrak B} \coloneqq \{((a_1, b_1), \ldots, (a_{r_i}, b_{r_i})) : (a_1, \ldots, a_{r_i}) \in R_i^\mathfrak A \hspace{6pt}\text{and}\hspace{6pt} (b_1, \ldots, b_{r_i}) \in R_i^\mathfrak B \}.\]
The following observations are easy to prove:
\begin{proposition}\hspace{0 cm}
\label{product rules}
    \begin{enumerate}[(i)]
        \item For any structure $H$ and natural number $r$ we have $(r\cdot \mathbf{1})\otimes H \cong r \cdot H$.
        \item For any natural number $n$ we have $\mathsf{C}_n \otimes \mathsf{C}_n \cong n \cdot \mathsf{C}_n$.
        \item For any structures $H, A,B$ we have $\hom(H, A\otimes B) = \hom(H,A)\cdot \hom(H, B)$.
    \end{enumerate}
\end{proposition}

We can now prove:

\homLemma*
\proof 
First we note that $\hom(A, n\cdot \mathbf{1}) = n^{c(A)}$. 
From Proposition \ref{product rules} we get $\mathsf{C}_n\otimes \mathsf{C}_n \cong n\cdot \mathsf{C}_n \cong (n\cdot \mathbf{1}) \otimes \mathsf{C}_n$. 
Therefore \begin{align*}
    \hom(A, \mathsf{C}_n) \cdot \hom(A, \mathsf{C}_n) 
    &= \hom(A, \mathsf{C}_n \otimes \mathsf{C}_n) \\
    &= \hom(A, n \cdot \mathsf{C}_n) \\
    &= \hom(A, n\cdot \mathbf{1}) \cdot \hom(A, \mathsf{C}_n),
\end{align*} so if $\hom(A, \mathsf{C}_n) \neq 0$ then $\hom(A, \mathsf{C}_n) = \hom(A, n\cdot \mathbf{1}) = n^{c(A)}$.
We then get \begin{align*}
    \hom(A, m\cdot \mathsf{C}_n) 
    &= \hom(A, (m\cdot \mathbf{1}) \otimes \mathsf{C}_n)\\
    &= \hom(A, m\cdot \mathbf{1}) \cdot \hom(A, \mathsf{C}_n)\\
    &= \begin{cases}
        0 &\If \hom(A, \mathsf{C}_n) = 0\\
        m^{c(A)}\cdot n^{c(A)} &\If \hom(A,\mathsf{C}_n) \neq 0
    \end{cases}\\
    &= \begin{cases}
            0 &\If n \nmid \gamma(A) \\
            (m\cdot n)^{c(A)} &\If n \mid \gamma(A)
        \end{cases}
\end{align*}
where the last step follows from Lemma \ref{girth-prop}
\qed

\section{Proof of Theorem \ref{n-ary-better-corollary}}
\label{proof-of-nary-thm}
Theorem \ref{better-corollary} shows that for every signature containing a binary relation, there exists a class of structures with that signature that is not decided by an adaptive left $k$-query algorithm over $\N$ for any $k$.
The key step in the proof of that result is Lemma \ref{hom-lemma}.
In this section, we extend this result to any signature that contains an $n$-ary relation with $n\geq2$.
We do this by emulating the symmetry of directed cycles to get an $n$-ary analogue to Lemma \ref{hom-lemma}.

\begin{definition}
    The $n$-ary cycle of length $d$ is defined with \begin{align*}
        \mathsf{C}_d^n \coloneqq (\{0,\ldots, d-1\}, R)
    \end{align*} where \begin{align*}
        R = \{(i \!\!\!\pmod d, \ldots, i+ (n-1) \!\!\!\pmod d) : i \in \{0,\ldots, d-1\}\}.
    \end{align*}
\end{definition}

For a structure $\A = (A, R)$ with one $n$-ary relation, we define $\A^* \coloneqq (A, R^*)$ where $R^*$ is the binary relation defined with: $R^*(a,b)$ holds if and only if there exists an $i$ and elements $a_1,\ldots,a_{i-1}, a_{i+2},\ldots, a_n$ such that \begin{align*}
    R(a_1,\ldots,a_{i-1}, a, b, a_{i+2}, \ldots, a_n).
\end{align*}
Note that $(\mathsf{C}_d^n)^*\cong \mathsf{C}_d$.
We can now prove:
\begin{lemma}
    \label{n-ary-transformation-lemma}
    Let $\A = (A, R)$ be a structure where $R$ is an $n$-ary relation. A function $f: \A \to \mathsf{C}_d^n$ is a homomorphism if and only if $f: \A^* \to \mathsf{C}_d$ is a homomorphism.
\end{lemma}
\proof 
First, assume $f$ is a homomorphism from $\A$ to $\mathsf{C}_d^n$ and let $R^*(a,b)$ be given. Then there exists $i$ and elements $a_j$ for $j\in \{1,\ldots, n\}\setminus\{i,i+1\}$ such that $R(a_1,\ldots,a_{i-1}, a, b, a_{i+2}, \ldots, a_n)$. Then we get \begin{align*}
        R(f(a_1),\ldots,f(a_{i-1}), f(a), f(b), f(a_{i+2}), \ldots, f(a_n))
\end{align*}
and therefore $R^*(f(a), f(b))$, so $f$ preserves $R^*$.\par 
For the other direction, assume $f$ preserves $R^*$ and let $R(a_1, \ldots, a_n)$ be given. For each $i \in \{1,\ldots, n-1\}$ we then have $R^*(a_i,a_{i+1})$, so $R^*(f(a_i), f(a_{i+1}))$ which means that $f(a_i)+1 = f(a_{i+1})\pmod d$. Since this holds for every $i$ we have \begin{align*}
    f(a_{j+1}) = f(a_1)+ j \!\!\! \pmod d
\end{align*} for each $j \in \{0,\ldots, n-1\}$. So indeed $R(f(a_1), \ldots, f(a_n))$ and we have shown that $f$ preserves $R$.
\qed

The following now immediately follows from Lemmas \ref{n-ary-transformation-lemma} and \ref{girth-prop}.
\begin{lemma}
    \label{n-ary-girth-lemma}
    Let $\A = (A, R)$ with $R$ $n$-ary. Then $\A \to \mathsf{C}_d^n$ if and only if $d\mid \gamma(\A^*)$.
\end{lemma}

We then obtain the analogue of Lemma \ref{hom-lemma} for structures with one $n$-ary relation.

\begin{lemma}
    \label{n-ary-hom-lemma}
    For every $\A= (A, R)$ with $R$ an $n$-ary relation and all positive integers $d,m$ we have \begin{align*}
        \hom(\A, m\cdot \mathsf{C}_d^n) = \begin{cases}
            0 &\If d \nmid \gamma(\A)\\
            (m\cdot d)^{c(\A)} &\If d\mid \gamma(\A)
        \end{cases}.
    \end{align*}
\end{lemma}
\proof 
In the same way as in Proposition \ref{product rules} we have that $\mathsf{C}_d^n \otimes \mathsf{C}_d^n \cong d \cdot \mathsf{C}_d^n$ by observing that the diagonal lines in the product form cycles.
Then we have $\mathsf{C}_d^n \otimes \mathsf{C}_d^n \cong d \cdot \mathsf{C}_d^n \cong (d\cdot \mathbf{1})\otimes \mathsf{C}_d^n$. Using this and Lemma \ref{n-ary-girth-lemma} the proof follows in the same way as the proof of Lemma \ref{hom-lemma}.
\qed
\par 
It is now easy to see that we can use the same proofs as before to prove analogues of Theorem \ref{main-theorem}, Corollary \ref{intermediate-corollary} and Theorem \ref{better-corollary}. 
\par 
The attentive reader might have noticed that the hypothesis $n\geq 2$ is not assumed in Lemma~\ref{n-ary-hom-lemma}.
Indeed, the lemma also holds in the case $n = 1$. 
However, the analogues of Theorem \ref{main-theorem}, Corollary \ref{intermediate-corollary}, and Theorem \ref{better-corollary} do not extend to this case. The reason is that if $n=1$ we have $\mathsf{C}_d^n \cong d \cdot \mathsf{C}_1^n$, whereas for all $n\geq 2$ and all $m>1$ we have that $\mathsf{C}_{m\cdot d}^n \ncong m\cdot \mathsf{C}_d^n$. Therefore we can only establish the analogues of these results for $n$-ary relations with $n\geq2$. In particular, we obtain:
\nAryBetterCorollary*

Finally, we also prove the claim that every class of structures with a signature containing only unary relation symbols is decided by a left bounded query algorithm over $\N$.
\begin{theorem}
\label{thm:unary-is-decidable}
    Every class of structures with signature $\tau = \{P_1, \ldots, P_k\}$, where each $P_i$ is a unary relation symbol, is decided by a non-adaptive left $2^k$-query algorithm over $\N$.
\end{theorem}
\begin{proof}
    For each subset $S \subseteq \tau$ let $F_S \coloneqq (\{0\}, (P_i^{F_S})_{i = 1}^k)$ where 
    \[P_i^{F_S} = \begin{cases}
        \{0\} &\If P_i \in S \\
        \varnothing &\If P_i \notin S
    \end{cases}.
    \]
    In words, $F_S$ is the singleton structure satisfying the predicates in $S$ and no others.
    Querying $\hom(F_S, A)$ for each $S$ gives information about the number of elements in $A$ satisfying each Boolean combination of the predicates in $\tau$.
    This information determines $A$ up to isomorphism and therefore suffices to classify $A$ correctly.
\end{proof}

\section{Proofs of Theorems About Left Query Algorithms over $\B$}
\label{appendix:proofs-from-unb-over-B}

Here, we first prove the correctness of the topological characterization of the expressive power of adaptive left unbounded query algorithms over $\B$.
\topologicalCharacterization*

\proof\hspace{0cm}
\begin{itemize}
    \item[$\Rightarrow$:] Let $G$ be an adaptive left unbounded query algorithm over $\B$ that decides $\C$. We want to show that $\C$ is clopen. Let $A \in \C$. Now $\sigma(A, G)$ is finite so there exist $A_0, \ldots , A_n$ such that \begin{align*}
        \sigma(A, G) = (\hom_\B(A_0, A) , \ldots, \hom_\B(A_n, A)).
    \end{align*}
    Let \begin{align*}
        U_A \coloneqq \bigcap_{A_i \to A} \up{A_i} \cap \bigcap_{A_i \nrightarrow A} \mathsf{FIN}(\tau) \setminus \up {A_i}.
    \end{align*}
    Then indeed $U_A$ is a basis element such that $A \in U_A$. Now for $A' \in U_A$ we have \begin{align*}
        (\hom_\B(A_0, A') , \ldots, \hom_\B(A_{n'}, A')) = (\hom_\B(A_0, A) , \ldots, \hom_\B(A_{n'}, A))
    \end{align*} for every $n'\leq n$. The computation path of $G$ on $A'$ is thus the same as the computation path on $A$, i.e. $\sigma(A', G) = \sigma(A, G)$.
    It therefore follows that the algorithm accepts $A'$ since it accepts $A$, and thus $U_A \subseteq \C$. We have now shown that $\C = \bigcup_{A \in \C} U_A$, so it is open. If $\C$ is decided by an adaptive left unbounded query algorithm over $\B$ then so is $\C^c$ (by switching \textsf{YES} and \textsf{NO} in deciding step of the algorithm). Our argument therefore also shows that $\C^c$ is open, so $\C$ is clopen.

    \item[$\Leftarrow$:] Let $\C \subseteq \mathsf{FIN}(\tau)$ be clopen. Let $A_0, A_1, A_2, \ldots$ be an enumeration of all structures with signature $\tau$ up to isomorphism. For any $\sigma = (\sigma_0, \ldots, \sigma_n) \in \B^{<\omega}$ we have a corresponding basis element\begin{align*}
        U_\sigma = \bigcap_{{\substack{i \leq n \\ \sigma_i = 1}}} \up{A_i} \cap \bigcap_{{\substack{i \leq n \\ \sigma_i = 0}}} \mathsf{FIN}(\tau)\setminus \up{A_i}.
    \end{align*}
    Moreover, if $A$ is the target structure defining $\sigma$, so \begin{align*}
        \sigma = (\hom_\B(A_0, A) , \ldots, \hom_\B(A_n, A))
    \end{align*}
    then we have that $A \in U_\sigma$.
    We now define the following algorithm: \begin{align*}
        G( \sigma) = \begin{cases}
            \mathsf{YES} &\If U_\sigma \subseteq \C \\
            \mathsf{NO} &\If U_\sigma \subseteq \C^c \\
            A_{\abs{\sigma}} &\hspace{6 pt}\text{otherwise}
        \end{cases}.
    \end{align*}
    If $\sigma = (\hom_\B(A_0, A) , \ldots, \hom_\B(A_n, A))$ and $U_\sigma \subseteq \C$ then clearly $A \in \C$. Similarly, if $U_\sigma \subseteq \C^c$ then $A \notin \C$ so the algorithm always classifies correctly when it terminates. To see that it always terminates we note that for any $A$ we have either $A \in \C$ or $A \in \C^c$. Assume, without loss of generality, that $A \in \C$. Then, since $\C$ is open, there is a basis element \begin{align*}
        U = \bigcap_{i = 0}^k \up{B_i} \cap \bigcap_{i= 0}^m \mathsf{FIN}(\tau)\setminus\up{D_i}
    \end{align*} such that $A \in U \subseteq \C$. Now, there exists an $n$ such that $B_i, D_j \in \{A_0, \ldots, A_n\}$ for all $i = 0,\ldots, k$ and $j = 0,\ldots, m$ (here we do not distinguish between isomorphic structures). 
    So, then indeed $A \in U_\sigma \subseteq U \subseteq \C$ for $\sigma = (\hom_\B(A_0, A) , \ldots, \hom_\B(A_n, A))$. This shows that the algorithm terminates. The argument is exactly the same for $A \in \C^c$. We have thus shown that $\C$ is decided by an adaptive left unbounded query algorithm over $\B$.
\end{itemize}
\qed
\par

We now show how to use the topological characterization to prove the following.
\datalogUnboundedLqaThm*
\proof 
Since $\C$ is homomorphism-closed it follows directly that $\C = \bigcup_{A \in \C} \up{A}$. Since $\up{A}$ is open for each $A$, it follows that $\C$ is open.
We now turn to proving (ii).
The right-to-left direction is immediate since $\bigcup_{j=0}^n \up{A_{i,j}}$ are closed sets.
For the other direction, we assume $\C$ is closed and thus clopen. 
Since $\C$ is closed and the basis is clopen we have that $\C$ is an intersection of basis elements:
\begin{align}
\label{cequation}
    \C = \bigcap_{i \in I} \qty(\bigcup_{j = 0}^{n_i} \up{B_{i,j}} \cup \bigcup_{j = 0}^{k_i} \mathsf{FIN}(\tau)\setminus\up{D_{i,j}}).
\end{align}
In fact, we can assume that $B_{i,j}$ is connected for each $(i,j)$:\par 
Each $B_{i,j}$ can be written as $B_{i,j} = A_1\oplus \ldots \oplus A_n$ where $A_k$ is connected for each $k$. Then $\up{B_{i,j}} = \bigcap_{k=1}^n \up{A_k}$ since $\oplus$ is a least upper bound operation in the lattice. Using the distributive law it is clear that each term $\bigcup_{j = 0}^{n_i} \up{B_{i,j}} \cup \bigcup_{j = 0}^{k_i} \mathsf{FIN}(\tau)\setminus\up{D_{i,j}}$ can be written in conjunctive normal form as: \begin{align*}
    \bigcap_{k = 0}^n \qty(\bigcup_{j = 0}^{n_i'} \up{A_{i,k,j}} \cup \bigcup_{j = 0}^{k_i'} \mathsf{FIN}(\tau)\setminus\up{D'_{i,k,j}})
\end{align*}
where each $A_{i,k,j}$ is connected. Thus we can write 
\begin{align*}
    \C =  \bigcap_{i \in I} \bigcap_{k = 0}^n \qty(\bigcup_{j = 0}^{n_i'} \up{A_{i,k,j}} \cup \bigcup_{j = 0}^{k_i'} \mathsf{FIN}(\tau)\setminus\up{D'_{i,k,j}}).
\end{align*} The outer intersections can be taken together to form an equation of the form (\ref{cequation}) where each $B_{i,j}$ is connected.
\par 
We can now assume equation (\ref{cequation}) holds and all $B_{i,j}$'s are connected. If \begin{align*}
     \C = \bigcap_{i \in I} \bigcup_{j = 0}^{n_i} \up{B_{i,j}}
\end{align*} then we are done. Otherwise there exists a structure $H\in \C$ such that $H \notin \bigcup_{j = 0}^{n_i} \up{B_{i,j}}$ for some $i$. Since $\C$ is homomorphism-closed we also have that $H' \coloneqq H\oplus \bigoplus_{j=0}^{k_i} D_{i,j} \in \C$, but $H' \notin \bigcup_{j = 0}^{k_i} \mathsf{FIN}(\tau)\setminus\up{D_{i,j}}$. This means that we must have $H' \in \bigcup_{j = 0}^{n_i} \up{B_{i,j}}$, so for some $j$ we have $B_{i,j}\rightarrow H'$. Then, since $B_{i,j}$ is connected and $B_{i,j} \nrightarrow H$, we must have $B_{i,j} \to D_{i,j'}$ for some $j'$. But then $\up{D_{i,j'}} \subseteq \up{B_{i,j}}$ so \begin{align*}
    \up{B_{i,j}} \cup (\mathsf{FIN}(\tau)\setminus\up{D_{i,j'}}) = \mathsf{FIN}(\tau)
\end{align*} and thus \begin{align*}
    \bigcup_{j = 0}^{n_i} \up{B_{i,j}} \cup \bigcup_{j = 0}^{k_i} \mathsf{FIN}(\tau)\setminus\up{D_{i,j}} =  \mathsf{FIN}(\tau). 
\end{align*}
We can therefore define 
\[I' \coloneqq \{ i \in I : \bigcup_{j = 0}^{n_i} \up{B_{i,j}} \cup \bigcup_{j = 0}^{k_i} \mathsf{FIN}(\tau)\setminus\up{D_{i,j}} \neq  \mathsf{FIN}(\tau)\}\]
and obtain \begin{align*}
    \C =  \bigcap_{i \in I'} \bigcup_{j = 0}^{n_i} \up{B_{i,j}}
\end{align*}
as desired.
\qed
\par 
We now use the topological characterization to prove Theorem \ref{thm:main-CSP}.
\mainCSP*
\begin{proof}
    If $\C$ is decided by a non-adaptive left query algorithm over $\B$ then it is clearly also decided by an adaptive one.
    For the other direction, first assume $\C = [A]_{\leftrightarrow}$ is decided by an adaptive left unbounded query algorithm over $\B$. Then $[A]_{\leftrightarrow}$ is clopen, so it is a union of basis elements. Since it is a singleton modulo homomorphic equivalence, it must be a union of a single basis element. It is thus an intersection of $k$ subbasis elements, which correspond to homomorphism queries. Thus it is clear that $[A]_{\leftrightarrow}$ is decided by a non-adaptive left $k$-query algorithm.
    \par
    Now assume $\C = \CSP(A)$ is decided by an adaptive left unbounded query algorithm over $\B$. Note that such an algorithm can be turned into an adaptive left unbounded query algorithm over $\B$ for $[A]_{\leftrightarrow}$ by also querying for the existence of a homomorphism from $A$. If such a homomorphism exists and the target structure is in $\CSP(A)$, then the algorithm accepts. From the above it therefore follows that $[A]_{\leftrightarrow}$ is also decided by a non-adaptive left query algorithm over $\B$. 
    Proposition 4.4 in ten Cate et al. \cite{ten2024homomorphism} states that $[A]_{\leftrightarrow}$ is decided by a non-adaptive left query algorithm over $\B$ if and only if $\CSP(A)$ is.
    Thus we have that $\CSP(A)$ is also decided by a non-adaptive left query algorithm over $\B$. 
\end{proof}

For the proof of the following theorem, we use an example of a $\CSP$ that is not first-order definable but has a simple Datalog program defining its complement.
\monadicDatalog*
\begin{proof}
    Let $\tau= \{R, P, Q\}$ be a signature with $R$ a binary relation while $P$ and $Q$ are unary. Define the structure $A\coloneqq \{ \{0,1\}, R^{A}, P^A, Q^A\}$ with $R^{A} \coloneqq \{(0,0), (1,1)\}$, $P_A = \{0\}$ and $Q_A=\{1\}$. A structure $B$ is an element of $\CSP(A)$ if and only if there is no undirected path from an element in $P^B$ to an element in $Q^B$. It is a standard application of the Ehrenfeucht-Fra\"{i}sse method to show that this is not first-order definable.
    It then follows from Theorem \ref{cor:csp-FO} that the complement of $\CSP(A)$ is not decided by an adaptive left unbounded query algorithm over $\B$.
    However, this class is defined by the following monadic linear Datalog program:
    \[
        \begin{array}{lll}
            X(x) &\leftarrow & P(x) \\
            X(y) & \leftarrow & X(x), R(x,y)\\
            X(y) & \leftarrow & X(x), R(y,x)\\
            Ans() & \leftarrow & X(y), Q(y)
        \end{array}
    \]
\end{proof}

\section{Proofs of Theorems About Right Query Algorithms}
\label{appendix:rqa-thm}
The goal of this section is to prove the theorems about right query algorithms. We begin with the following:

\rightQueryTheorem*
This theorem follows almost immediately from the following result, which is a variation of a result by Wu and Kolaitis (Theorem 5.38 in \cite{wu2023study}). The proof is exactly the same as their proof so we leave it out.

\begin{lemma}
\label{Wu}
    For every $n>0$ and every signature $\tau$, there exists a structure $F(n,\tau)$ of signature $\tau$ such that for every two structures $H,H'$ with signature $\tau$ and of size $n$ we have \begin{align*}
        H\cong H' \text{ if and only if } \hom(H, F(n,\tau)) = \hom(H', F(n,\tau)).
    \end{align*}
\end{lemma}

\par 
\begin{proof}[Proof of Theorem \ref{right-query-theorem}]
Let $\tau$ be a signature. Let $\mathfrak{2}_\tau$ be the complete structure with domain of size 2 and signature $\tau$. Concretely, $\mathfrak{2}_\tau \coloneqq (\{0,1\}, (R_\mathfrak{2})_{R \in \tau})$ where \begin{align*}
    R_\mathfrak{2} = \{0,1\}^{m_R}
\end{align*} where $m_R$ is the arity of $R$. \par 
Then, for any structure $H$ with signature $\tau$ we have \begin{align*}
    \hom(H, \mathfrak{2}_\tau) = 2^\abs{H}.
\end{align*}
Our algorithm thus runs as follows:
Given input $H$, query $\hom(H, \mathfrak{2}_\tau)$ to obtain $\abs{H}$. Then query $\hom(H, F(\abs{H}, \tau))$. 
\par 
By Lemma \ref{Wu} this determines the structure $H$ up to isomorphism. 
Thus every input $A$ has a unique computation path $\sigma(A)$ up to this point. We can thus set 
\[G(\sigma(A)) \coloneqq \begin{cases}
    \textsf{YES} &\If A \in \C \\
    \textsf{NO} & \If A \notin \C
\end{cases}\]
where $\C$ is the class we are trying to decide.
\end{proof}

\medskip\par 
We now turn towards proving Theorem \ref{unbounded-helps-rqa}.
\unboundedHelpsRqa*

Before we prove the theorem we state a few lemmas.

\begin{lemma}[{Proposition 1.13 in \cite{hell2004graphs}}]
    \label{path-lemma} 
    A digraph $H$ has no oriented cycle with non-zero net length if and only if $H \to P_{\abs{H}-1}$.
\end{lemma}

\begin{lemma}
\label{rqa-unb-lemma1}
Let $A$ be a finite directed graph. We have that $A$ has no oriented cycle with non-zero net length if and only if there exists $n>0$ such that $\hom(A, \mathsf{P}_n) >0$.
\end{lemma}
\proof \hspace{0cm}
\begin{itemize}
    \item[$\Rightarrow$:] Since the net length of oriented walks is preserved under homomorphisms and the net length of the longest oriented walk in $\mathsf{P}_n$ is $n$, we see that if $A \to \mathsf{P}_n$ then $A$ has no oriented walk with net length $>n$. Then indeed $A$ has no oriented cycle with positive net length (and thus no oriented cycle with non-zero net length). 
    \item[$\Leftarrow$:] If $A$ has no oriented cycles with positive net length then it follows from Lemma \ref{path-lemma} that $A \to P_{\abs{A}-1}$.
\end{itemize}
\qed

\begin{lemma} 
\label{rqa-unb-lemma2}
Let $A$ be a finite directed graph. We have that $A$ has an oriented cycle with non-zero net length if and only if there exists $n>0$ such that $\hom(A, \mathsf{C}_n)=0$.
\end{lemma}
\proof\hspace{0cm}
\begin{itemize}
    \item[$\Rightarrow$:] If $\hom(A, \mathsf{C}_n) = 0$, then by Lemma \ref{girth-prop} we have $n \nmid \gamma(A)$. In particular $\gamma(A) \neq 0$, so $A$ has an oriented cycle with non-zero net length. \par 
    \item[$\Leftarrow$:] If $A$ has a oriented cycle of non-zero net length then $\gamma(A) \neq 0$. But then $\gamma(A) + 1 \nmid \gamma(A)$ so by Lemma \ref{girth-prop} we have $\hom(A, \mathsf{C}_{\gamma(A)+1}) = 0$.
\end{itemize}
\qed

Surprisingly, essentially the same unbounded algorithm as we used to prove Theorem \ref{leftUnboundedMoreExpressive} works here.

\begin{proof}[Proof of Theorem \ref{unbounded-helps-rqa}]
We first show that the class $\C$ of all digraphs that have an oriented cycle with non-zero net length can be decided by an adaptive right unbounded query algorithm over $\B$.
The algorithm runs as follows:
Let $A$ be the given input.
For $n = 1,2, \ldots$ query $\hom_\B(A, \mathsf{P}_n)$ and $\hom_\B(A, \mathsf{C}_n)$. If at any point $\hom_\B(A, \mathsf{P}_n) = 1$ then halt and output \textsf{NO}. If at any point $\hom_\B(A, \mathsf{C}_n) = 0$ then halt and output \textsf{YES}. \par 
It follows from Lemmas \ref{rqa-unb-lemma1} and \ref{rqa-unb-lemma2} that on every input, this algorithm halts and produces the correct outcome.

\smallskip \par 
We now show that the class $\C$ cannot be decided by a non-adaptive right $k$-query algorithm over $\B$ and thus not by an adaptive right $k$-query algorithm over $\B$. \par 
Let $(\mathcal F, X)$ be a non-adaptive right $k$-query algorithm. Let 
\[m \coloneqq \max (\{ \abs{F} : F \in \mathcal F\}),\]
so $m\geq \abs{F}$ for every $F \in \mathcal F$. Then, since the lengths of directed walks are preserved under homomorphism, it holds that for any $F \in \mathcal F$ we have $P_{m} \to F$ if and only if $F$ has a directed cycle. 
Also, if $F$ has a directed cycle, then that cycle has length $\leq m$ so $\mathsf{C}_{m!} \to F$. It is also clear that if $\mathsf{C}_{m!} \to F$ then $F$ must have a directed cycle. Thus $\mathsf{C}_{m!} \to F$ if and only if $F$ has a directed cycle. We have thus shown that for $F \in \mathcal{F}$ we have\begin{align*}
    P_{m} \to F \quad \text{if and only if}\quad \mathsf{C}_{m!} \to F.
\end{align*} So the class $\C'$ that the algorithm decides either contains both $P_{m}$ and $\mathsf{C}_{m!}$ or neither of them. In both cases $\C' \neq \C$ so the algorithm does not decide the class $\C$.
\end{proof}

The class from the above result is Datalog definable:
\begin{theorem}
    \label{thm:non-zero-net-length-is-Datalog}
    The class of all digraphs that have an oriented cycle with non-zero net length is Datalog definable.
\end{theorem}
\begin{proof}
    A Datalog program defining this class is the following:
    \[\begin{array}{ccc}
        X(a,b) & \leftarrow & a=b \\
        X(a,b) & \leftarrow & X(a', b'), R(a', a), R(b',b) \\
        X(a,b) & \leftarrow & X(a', b'), R(a, a'), R(b, b')\\
        X(a,b) & \leftarrow & X(a,c), X(c,b)\\
        \\
        Y(a,b) & \leftarrow & R(a, b)\\
        Y(a,b) & \leftarrow & Y(a, c), X(c, b) \\
        Y(a,b) & \leftarrow & X(a, c), Y(c, b)\\
        Y(a,b) & \leftarrow & Y(a, c), Y(c, b)\\
        \\
        Ans() & \leftarrow & Y(a, a)    \\
    \end{array}\]
    To see that this program does indeed define the class of digraphs that have an oriented cycle with non-zero net length, we begin by showing that the IDB $X$ defines the relation \begin{align}
    \label{eq:X-fact}
        X(a,b) \iff \text{there is an oriented walk from $a$ to $b$ with net length 0}.
    \end{align}
    The left-to-right direction is clear. To see the other direction, we use induction on the length of walks of net length 0.
    Let $P$ be any walk of length $2n$ but net length 0. 
    We have a few cases: \begin{itemize}
        \item If $P$ is of the form $(a_0, -, a_1, r_1, \ldots, a_{2n-1}, +, a_{2n})$ then $(a_1, r_1, \ldots, a_{2n-1})$ is an oriented walk of length $2(n-1)$ but net length 0, so by the induction hypothesis $X(a_1, a_{2n-1})$. It then follows from the second rule in the program that $X(a_0, a_{2n})$. 
        \item The case where $P$ is of the form $(a_0, +, a_1, r_1, \ldots, a_{2n-1}, -, a_{2n})$ is similar.
        \item Now assume $P = (a_0, +, a_1, r_1, \ldots, a_{2n-1}, +, a_{2n})$. Since it has net length 0 there must exist subwalks $Q$, $R$ and $Z$ of net length 0 such that \begin{align*}
            P = (a_0, +, a_1) \circ Q \circ (a_{i}, -, a_{i+1}) \circ R \circ (a_j, - , a_{j+1}) \circ Z \circ (a_{2n-1}, + , a_{2n}).
        \end{align*}
        (Namely, we pick $i$ as the smallest number such that $(a_0, +, \ldots, a_{i+1})$ is an oriented walk of net length 0 and we pick $j$ as the smallest number such that $(a_0, +, \ldots, a_{j+1})$ is an oriented walk of net length $-1$). By the induction hypothesis, $X(a_1, a_i)$, $X(a_{i+1}, a_j)$ and $X(a_{j+1} , a_{2n-1})$ hold. Using the rules, it is clear that we also have $X(a_0, a_{2n})$.
        \item The case where $P$ is of the form $(a_0, -, a_1, r_1, \ldots, a_{2n-1}, -, a_{2n})$ is similar.
    \end{itemize}
    We have now proved the bi-implication (\ref{eq:X-fact}). It is then easy to see that $Y(a,b)$ holds if and only if there is an oriented walk of positive net length from $a$ to $b$. Now we see that the program accepts if and only if there is a closed oriented walk of positive net length, which holds if and only if there is an oriented cycle of non-zero net length (this follows from our proof of Proposition \ref{girth-prop}).
\end{proof}

We now prove Theorem \ref{thm:rqa-unb-does-not-help}. We do it by proving a sequence of equivalences.
\begin{theorem}
    The following are equivalent:
    \begin{enumerate}[(i)]
        \item $[A]_{\leftrightarrow}$ is decided by an adaptive right unbounded query algorithm over $\B$.
        \item $\up{A}$ is decided by an adaptive right unbounded query algorithm over $\B$.
        \item $[A]_{\leftrightarrow}$ is decided by a non-adaptive right query $k$-algorithm over $\B$ for some $k$.
        \item $\up{A}$ is decided by a non-adaptive right $k$-query algorithm over $\B$ for some $k$.
        \item $A$ is homomorphically equivalent to an acyclic structure.
    \end{enumerate}
\end{theorem}
\begin{proof}
    The equivalence of (iii), (iv), and (v) is Theorem 6.5 in ten Cate et al. \cite{ten2024homomorphism}. 
    It is easy to see that (iii) implies (i) and that (iv) implies (ii).
    Now, if $[A]_{\leftrightarrow}$ is clopen then it is a union of basis elements. Since it is a singleton modulo homomorphic equivalence, it must be a union of a single basis element. It is thus an intersection of $k$ subbasis elements, which correspond to right homomorphism queries. Thus it is clear that $[A]_{\leftrightarrow}$ is decided by a non-adaptive right $k$-query algorithm. This shows that (i) implies (iii).
    Now assume $\up{A}$ is decided by an adaptive right unbounded query algorithm over $\B$. Note that such an algorithm can be turned into an adaptive right unbounded query algorithm over $\B$ for $[A]_{\leftrightarrow}$ by also querying for the existence of a homomorphism to $A$. If such a homomorphism exists and the target structure is in $\up{A}$, then the algorithm accepts. This shows that (ii) implies (i).
\par 
    We have thus shown that all of the conditions are equivalent.
\end{proof}

\end{document}